\documentclass[10pt]{article}

\usepackage[T1]{fontenc}
\usepackage[utf8]{inputenc}
\usepackage[a4paper,total={7in,10in}]{geometry}
\usepackage[dvipsnames,table,xcdraw]{xcolor}
\usepackage{amsmath,amssymb,amsfonts,amsthm}
\usepackage{aliascnt}
\usepackage{graphicx}
\usepackage{subcaption}
\usepackage{float}
\usepackage{array}
\usepackage{verbatim}
\usepackage{comment}
\usepackage{tikz}
\usepackage{tkz-berge}
\usepackage{rotating}
\usepackage{enumitem}
\usepackage[pagewise]{lineno}
\usepackage[hidelinks]{hyperref}
\hypersetup{
  pdftitle={Complexity, approximation, and extension of proper \{a,b\}-edge-weightings},
  pdfauthor={Péter Madarasi and Máté Simon},
  pdfsubject={Decision, approximation, and extension problems for \{a,b\}-edge-weightings},
  pdfkeywords={1-2-3 conjecture, neighbor-sum-distinguishing edge-weightings,
               \{a,b\}-edge-weightings, locally irregular edge-colorings,
               planar graphs, approximation algorithms,
               extension of partial edge-weightings}
}
\usepackage[nameinlink,noabbrev,capitalise,sort&compress]{cleveref}

\def\Z{\mathbb Z}

\def\1{\mathbb 1}

\DeclareMathOperator{\dist}{dist}
\DeclareMathOperator{\OPT}{OPT}

\usepackage[textsize=tiny]{todonotes}

\numberwithin{equation}{section}

\makeatletter
\edef\ReferenceTool@CleverefVersion{\csname ver@cleveref.sty\endcsname}%
\def\ReferenceTool@BrokenCleverefVersion{2018/03/27 v0.21.4 Intelligent cross-referencing}%
\ifx\ReferenceTool@CleverefVersion\ReferenceTool@BrokenCleverefVersion
  \let\cpageref\relax
  \DeclareRobustCommand{\cpageref}{%
    \@ifstar{\@crefstar{cpageref}}{\@cref{cpageref}}}%
  \let\Cpageref\relax
  \DeclareRobustCommand{\Cpageref}{%
    \@ifstar{\@crefstar{Cpageref}}{\@cref{Cpageref}}}%
  \providecommand*{\@setcpagerefrange}[3]{%
    \@@setcpagerefrange{#1}{#2}{cref}{#3}}%
  \providecommand*{\@setCpagerefrange}[3]{%
    \@@setcpagerefrange{#1}{#2}{Cref}{#3}}%
  \providecommand*{\@setlabelcpagerefrange}[3]{%
    \@@setcpagerefrange{#1}{#2}{labelcref}{#3}}%
\fi
\makeatother

\AtBeginDocument{%
}

\newcommand{\DeclareEquationCrefFormat}[1]{%
  \crefformat{#1}{##2\textup{(##1)}##3}%
  \Crefformat{#1}{##2\textup{(##1)}##3}%
  \crefrangeformat{#1}{##3\textup{(##1)}##4\nobreakdash--##5\textup{(##2)}##6}%
  \Crefrangeformat{#1}{##3\textup{(##1)}##4\nobreakdash--##5\textup{(##2)}##6}%
  \crefmultiformat{#1}%
    {##2\textup{(##1)}##3}%
    { and~##2\textup{(##1)}##3}%
    {, ##2\textup{(##1)}##3}%
    {, and~##2\textup{(##1)}##3}%
  \Crefmultiformat{#1}%
    {##2\textup{(##1)}##3}%
    { and~##2\textup{(##1)}##3}%
    {, ##2\textup{(##1)}##3}%
    {, and~##2\textup{(##1)}##3}%
  \crefrangemultiformat{#1}%
    {##3\textup{(##1)}##4\nobreakdash--##5\textup{(##2)}##6}%
    { and~##3\textup{(##1)}##4\nobreakdash--##5\textup{(##2)}##6}%
    {, ##3\textup{(##1)}##4\nobreakdash--##5\textup{(##2)}##6}%
    {, and~##3\textup{(##1)}##4\nobreakdash--##5\textup{(##2)}##6}%
  \Crefrangemultiformat{#1}%
    {##3\textup{(##1)}##4\nobreakdash--##5\textup{(##2)}##6}%
    { and~##3\textup{(##1)}##4\nobreakdash--##5\textup{(##2)}##6}%
    {, ##3\textup{(##1)}##4\nobreakdash--##5\textup{(##2)}##6}%
    {, and~##3\textup{(##1)}##4\nobreakdash--##5\textup{(##2)}##6}%
  \labelcrefformat{#1}{##2\textup{(##1)}##3}%
  \labelcrefrangeformat{#1}{##3\textup{(##1)}##4\nobreakdash--##5\textup{(##2)}##6}%
  \labelcrefmultiformat{#1}%
    {##2\textup{(##1)}##3}%
    { and~##2\textup{(##1)}##3}%
    {, ##2\textup{(##1)}##3}%
    {, and~##2\textup{(##1)}##3}%
  \labelcrefrangemultiformat{#1}%
    {##3\textup{(##1)}##4\nobreakdash--##5\textup{(##2)}##6}%
    { and~##3\textup{(##1)}##4\nobreakdash--##5\textup{(##2)}##6}%
    {, ##3\textup{(##1)}##4\nobreakdash--##5\textup{(##2)}##6}%
    {, and~##3\textup{(##1)}##4\nobreakdash--##5\textup{(##2)}##6}%
}
\DeclareEquationCrefFormat{equation}
\DeclareEquationCrefFormat{subequation}

\newcommand{\DeclareItemCrefFormat}[1]{%
  \crefformat{#1}{##2\textup{##1}##3}%
  \Crefformat{#1}{##2\textup{##1}##3}%
  \crefrangeformat{#1}{##3\textup{##1}##4\nobreakdash--##5\textup{##2}##6}%
  \Crefrangeformat{#1}{##3\textup{##1}##4\nobreakdash--##5\textup{##2}##6}%
  \crefmultiformat{#1}%
    {##2\textup{##1}##3}{ and~##2\textup{##1}##3}%
    {, ##2\textup{##1}##3}{, and~##2\textup{##1}##3}%
  \Crefmultiformat{#1}%
    {##2\textup{##1}##3}{ and~##2\textup{##1}##3}%
    {, ##2\textup{##1}##3}{, and~##2\textup{##1}##3}%
  \crefrangemultiformat{#1}%
    {##3\textup{##1}##4\nobreakdash--##5\textup{##2}##6}%
    { and~##3\textup{##1}##4\nobreakdash--##5\textup{##2}##6}%
    {, ##3\textup{##1}##4\nobreakdash--##5\textup{##2}##6}%
    {, and~##3\textup{##1}##4\nobreakdash--##5\textup{##2}##6}%
  \Crefrangemultiformat{#1}%
    {##3\textup{##1}##4\nobreakdash--##5\textup{##2}##6}%
    { and~##3\textup{##1}##4\nobreakdash--##5\textup{##2}##6}%
    {, ##3\textup{##1}##4\nobreakdash--##5\textup{##2}##6}%
    {, and~##3\textup{##1}##4\nobreakdash--##5\textup{##2}##6}%
  \labelcrefformat{#1}{##2\textup{##1}##3}%
  \labelcrefrangeformat{#1}{##3\textup{##1}##4\nobreakdash--##5\textup{##2}##6}%
  \labelcrefmultiformat{#1}%
    {##2\textup{##1}##3}{ and~##2\textup{##1}##3}%
    {, ##2\textup{##1}##3}{, and~##2\textup{##1}##3}%
  \labelcrefrangemultiformat{#1}%
    {##3\textup{##1}##4\nobreakdash--##5\textup{##2}##6}%
    { and~##3\textup{##1}##4\nobreakdash--##5\textup{##2}##6}%
    {, ##3\textup{##1}##4\nobreakdash--##5\textup{##2}##6}%
    {, and~##3\textup{##1}##4\nobreakdash--##5\textup{##2}##6}%
}
\DeclareItemCrefFormat{enumi}
\DeclareItemCrefFormat{enumii}
\DeclareItemCrefFormat{enumiii}
\DeclareItemCrefFormat{enumiv}

\theoremstyle{plain}
\newtheorem{theorem}{Theorem}[section]

\newaliascnt{lemma}{theorem}
\newtheorem{lemma}[lemma]{Lemma}
\aliascntresetthe{lemma}

\newaliascnt{proposition}{theorem}
\newtheorem{proposition}[proposition]{Proposition}
\aliascntresetthe{proposition}

\newaliascnt{corollary}{theorem}
\newtheorem{corollary}[corollary]{Corollary}
\aliascntresetthe{corollary}

\newaliascnt{claim}{theorem}
\newtheorem{claim}[claim]{Claim}
\aliascntresetthe{claim}

\newaliascnt{prob}{theorem}

\aliascntresetthe{prob}

\newaliascnt{conjecture}{theorem}

\aliascntresetthe{conjecture}

\theoremstyle{definition}
\newaliascnt{definition}{theorem}

\aliascntresetthe{definition}

\theoremstyle{remark}
\newaliascnt{remark}{theorem}

\aliascntresetthe{remark}

\newcommand{\DeclareNamedCrefType}[3]{%
  \crefname{#1}{#2}{#3}%
  \Crefname{#1}{#2}{#3}%
  \crefrangelabelformat{#1}{##3##1##4\nobreakdash--##5##2##6}%
}

\DeclareNamedCrefType{section}{Section}{Sections}
\DeclareNamedCrefType{subsection}{Section}{Sections}
\DeclareNamedCrefType{subsubsection}{Section}{Sections}
\DeclareNamedCrefType{appendix}{Appendix}{Appendices}
\DeclareNamedCrefType{figure}{Figure}{Figures}
\DeclareNamedCrefType{table}{Table}{Tables}
\DeclareNamedCrefType{footnote}{Footnote}{Footnotes}
\DeclareNamedCrefType{page}{Page}{Pages}

\DeclareNamedCrefType{theorem}{Theorem}{Theorems}
\DeclareNamedCrefType{lemma}{Lemma}{Lemmas}
\DeclareNamedCrefType{proposition}{Proposition}{Propositions}
\DeclareNamedCrefType{corollary}{Corollary}{Corollaries}
\DeclareNamedCrefType{claim}{Claim}{Claims}
\DeclareNamedCrefType{prob}{Problem}{Problems}
\DeclareNamedCrefType{conjecture}{Conjecture}{Conjectures}
\DeclareNamedCrefType{definition}{Definition}{Definitions}
\DeclareNamedCrefType{remark}{Remark}{Remarks}

\makeatletter
\@onlypreamble\DeclareEquationCrefFormat
\@onlypreamble\DeclareItemCrefFormat
\@onlypreamble\DeclareNamedCrefType
\makeatother

\SetVertexNormal[FillColor=gray!25]
\newcommand{\figureedgelength}{1.15cm}
\tikzset{
  figure vertex/.style={circle,draw,fill=gray!25,minimum size=8pt,inner sep=0pt},
  figure edge/.style={line width=0.4pt}
}
\newcommand{\drawunitrhombus}[3]{%
  \node[figure vertex] (#1-left) at ($(#2)!{1/(2+sqrt(3))}!(#3)$) {};
  \node[figure vertex] (#1-right) at ($(#3)!{1/(2+sqrt(3))}!(#2)$) {};
  \coordinate (#1-mid) at ($(#1-left)!0.5!(#1-right)$);
  \node[figure vertex] (#1-upper) at ($(#1-mid)!{1/sqrt(3)}!90:(#1-right)$) {};
  \node[figure vertex] (#1-lower) at ($(#1-mid)!{1/sqrt(3)}!-90:(#1-right)$) {};
  \draw[figure edge] (#2)--(#1-left)--(#1-upper)--(#1-right)--(#3);
  \draw[figure edge] (#1-left)--(#1-lower)--(#1-right);
  \draw[figure edge] (#1-upper)--(#1-lower);
}

\usetikzlibrary{decorations.pathreplacing,calc}
\tikzset{draw half paths/.style 2 args={%
decoration={show path construction,
lineto code={
\draw [#1] (\tikzinputsegmentfirst) -- ($(\tikzinputsegmentfirst)!0.15!(\tikzinputsegmentlast)$);
\draw [#2] ($(\tikzinputsegmentfirst)!0.15!(\tikzinputsegmentlast)$) -- (\tikzinputsegmentlast);
}
}, decorate
}}

\begin{document}

\title{Complexity, approximation, and extension of\\ proper $\{a,b\}$-edge-weightings}

\author{%
  P\'eter Madarasi\thanks{HUN-REN Alfr\'ed R\'enyi Institute of Mathematics, Re\'altanoda u.\ 13--15., Budapest H-1053, Hungary; and Department of Operations Research, ELTE E\"otv\"os Lor\'and University, P\'azm\'any P.\ s.\ 1/c, Budapest H-1117, Hungary. E-mail: \texttt{madarasi@renyi.hu}, corresp.\ author}
\and
M\'at\'e Simon\thanks{Department of Operations Research, ELTE E\"otv\"os Lor\'and University, P\'azm\'any P.\ s.\ 1/c, Budapest H-1117, Hungary. E-mail: \texttt{ghlkwp@student.elte.hu}}%
}

\date{\vspace{-5mm}}

\maketitle

\begin{abstract}
For distinct integers $a$ and $b$, an \emph{$\{a,b\}$-edge-weighting} assigns $a$ or $b$ to each edge and labels each vertex by the sum of its incident weights.
Such a weighting is \emph{proper} if adjacent vertices receive distinct labels.
We prove that, for every fixed pair of distinct integers, deciding whether a proper weighting exists is NP-complete even for simple cubic planar graphs.
On planar multigraphs with $m$ edges, we give an exact $2^{O(\sqrt m)}$-time algorithm and, assuming the Exponential Time Hypothesis (ETH), exclude $2^{o(\sqrt m)}$-time algorithms even for simple cubic planar graphs.
As a consequence, locally irregular $2$-edge-coloring is NP-complete on simple cubic planar graphs, admits a deterministic $2^{O(\sqrt n)}$-time algorithm on $n$-vertex graphs in this class, and admits no $2^{o(\sqrt n)}$-time algorithm under ETH.
For maximizing the number of edges joining vertices with distinct labels, we give a deterministic efficient polynomial-time approximation scheme (EPTAS) on planar multigraphs, a polynomial-time $1/2$-approximation on multigraphs, and APX-completeness even on simple cubic graphs.
Extending a partial $\{a,b\}$-edge-weighting to a proper one is NP-complete for every fixed pair even on simple cubic planar bipartite graphs, while it is polynomial-time solvable on trees.
The hardness persists even when the prescribed edges form disjoint paths of length $6$ and all edges of each path have the same prescribed weight.

\end{abstract}

\section{Introduction}

Unless stated otherwise, $G=(V,E)$ denotes a finite, simple, undirected graph.
For a vertex $v\in V(G)$, let $E_G(v)$ denote the set of edges incident with $v$, and let $d_G(v)=|E_G(v)|$.
For a set of integers $W\subseteq\Z$, an assignment $w:E(G)\rightarrow W$ is called a \emph{$W$-edge-weighting}.
The label induced by $w$ at a vertex $v$ is $z_w(v)=\sum_{e\in E_G(v)}w(e)$; when the weighting is clear, we write $z(v)$.
A $W$-edge-weighting is \emph{proper} if $z_w(u)\neq z_w(v)$ for every edge $uv\in E(G)$.
An edge is \emph{proper} under $w$ if its endpoints receive distinct labels, and \emph{improper} otherwise.
An isolated edge is improper under every edge-weighting, because both endpoints receive the weight of that edge as their label.
If $G$ has a proper $W$-edge-weighting, then we say that $G$ has the \emph{$W$-property}.

In 2004, Karo\'{n}ski, \L{}uczak, and Thomason proposed the $1$-$2$-$3$ conjecture, which states that every simple graph without isolated edges has the $\{1,2,3\}$-property~\cite{karonski2004edge}.
Throughout the paper, $a$ and $b$ denote fixed distinct integers.
This paper studies decision, optimization, and extension problems for $\{a,b\}$-edge-weightings motivated by that conjecture.
A graph is \emph{locally irregular} if adjacent vertices have distinct degrees.
A \emph{locally irregular $2$-edge-coloring} is a $2$-edge-coloring in which the edges of each color induce a locally irregular subgraph.
On regular graphs, locally irregular $2$-edge-colorings are equivalent to proper $\{a,b\}$-edge-weightings; see \cref{prop:regular-weight-normalization} and~\cite{baudon2015decomposing}.

\subsection{Our results}

First, for every fixed pair of distinct integers $a$ and $b$, \cref{thm:planar-ab-weighting} proves that deciding whether a proper $\{a,b\}$-edge-weighting exists is NP-complete even for simple cubic planar graphs.
On connected loopless planar multigraphs with $m$ edges, \cref{thm:planar-exact} gives an exact deterministic algorithm with running time $2^{O(\sqrt m)}$.
Assuming the Exponential Time Hypothesis, \cref{thm:planar-eth-lower-bound} shows that the square-root dependence in the exponent is optimal even for connected simple cubic planar graphs.
For $\{1,2\}$-edge-weightings, the NP-completeness result strengthens the result of Dehghan, Sadeghi, and Ahadi for simple cubic graphs by also requiring planarity~\cite[Theorem~1]{dehghan2013algorithmic}.
\Cref{cor:liec-planar-cubic-hardness,cor:liec-planar-cubic-subexp} show that locally irregular $2$-edge-coloring is NP-complete even for simple cubic planar graphs.
On $n$-vertex simple cubic planar graphs, the problem is solvable in deterministic time $2^{O(\sqrt n)}$ and, assuming the Exponential Time Hypothesis, admits no $2^{o(\sqrt n)}$-time algorithm.
Previously, NP-completeness was known for planar graphs~\cite[Theorem~3]{ahadi2018decomposing} and, separately, for cubic graphs via the equivalent notion of $2$-detectability~\cite{havet2014detection}.

Second, we consider the problem of maximizing the number of proper edges.
For every fixed pair of distinct integers $a$ and $b$, \cref{thm:planar-eptas} gives a deterministic efficient polynomial-time approximation scheme (EPTAS) on connected loopless planar multigraphs.
On loopless multigraphs, \cref{thm:half-approx-ew} gives a polynomial-time $1/2$-approximation for every fixed pair of distinct integers $a$ and $b$.
For each fixed pair, \cref{thm:apx-complete-cubic} proves that the maximization problem is APX-complete even on simple cubic graphs.

Third, we study whether a partial $\{a,b\}$-edge-weighting, which prescribes the weights on a subset of the edges, can be extended to a proper weighting.
For every fixed pair of distinct integers $a$ and $b$, \cref{thm:extProbHard} proves that this extension problem is NP-complete even on simple cubic planar bipartite graphs.
Moreover, NP-hardness holds even when each component of the subgraph formed by the prescribed edges is a path of length $6$, all of whose edges have the same prescribed weight.
For locally irregular $2$-edge-colorings, \cref{cor:liec-extension-hardness} proves that deciding whether a partial $2$-edge-coloring extends to a locally irregular one is NP-complete on simple cubic planar bipartite graphs under the same restriction on the prescribed edges.
This contrasts with the fact that every cubic bipartite graph admits a locally irregular $2$-edge-coloring~\cite{baudon2015decomposing}.
\Cref{thm:treet} gives a polynomial-time algorithm for the extension problem on trees for every fixed pair of distinct integers $a$ and $b$.
Together, \cref{thm:treet,thm:antifactor} give an alternative polynomial-time algorithm for the antifactor problem on trees.

\subsection{Motivation and previous results}

The study of neighbor-sum-distinguishing edge-weightings originates in graph irregularity.
No simple graph on at least two vertices can have pairwise distinct degrees at all vertices, because an $n$-vertex graph with that property would have both a vertex of degree $0$ and a vertex of degree $n-1$.
Chartrand et al.\ therefore considered replacing each edge by parallel copies so that all vertex degrees become pairwise distinct~\cite{chartrand1988irregular}.
For a simple graph $G$ with no isolated edges and at most one isolated vertex, the \emph{irregularity strength} $s(G)$ is the least positive integer $k$ for which there is an edge-weighting $w:E(G)\rightarrow\{1,\dots,k\}$ such that the labels $z_w(v)$ are pairwise distinct over all $v\in V(G)$.
Further results on irregularity strength appear in~\cite{aigner1990irregular,faudree1987bound,nierhoff2000tight}.

A weaker requirement asks only that adjacent vertices have distinct degrees after the edge replacements.
If an edge $e$ is replaced by $w(e)$ parallel copies, then the resulting degree of a vertex $v$ is $\sum_{e\in E_G(v)}w(e)$.
Thus allowing between $1$ and $k$ copies of each edge is equivalent to seeking a proper $\{1,\dots,k\}$-edge-weighting.
Restricting the allowed weights to two integers gives the $\{a,b\}$-edge-weighting problem studied here.
Dudek and Wajc proved that deciding the $\{1,2\}$-property is NP-complete~\cite{dudek2011complexity}, and Dehghan, Sadeghi, and Ahadi strengthened this result to simple cubic graphs~\cite[Theorem~1]{dehghan2013algorithmic}.
On a regular graph, replacing one pair of distinct edge weights by another transforms every vertex label by the same affine function and therefore preserves properness; see \cref{prop:regular-weight-normalization}.
Thus the NP-completeness result for simple cubic graphs extends to every fixed pair of distinct integers $a$ and $b$.
Bensmail also gave a direct proof of NP-completeness for every pair of distinct real weights on general graphs in an unpublished manuscript~\cite{bensmail2013vertex}.

For regular graphs, the equivalence between proper $\{a,b\}$-edge-weightings and locally irregular $2$-edge-colorings connects the present problem with decompositions into locally irregular subgraphs~\cite{baudon2015decomposing}.
Baudon et al.\ proved that every regular bipartite graph of minimum degree at least $3$ admits a locally irregular $2$-edge-coloring.
Using the equivalence on regular graphs, the result of Havet, Paramaguru, and Sampathkumar implies that deciding whether a cubic graph admits a locally irregular $2$-edge-coloring is NP-complete~\cite{havet2014detection}.
Baudon, Bensmail, and Sopena proved that deciding whether a graph admits a locally irregular $2$-edge-coloring is NP-complete in general~\cite{baudon2015complexity}.
Ahadi et al.\ later gave a planar NP-completeness proof~\cite[Theorem~3]{ahadi2018decomposing}.
More recently, Lu\v{z}ar et al.\ studied when two colors suffice on cubic graphs and exhibited an infinite family of cubic graphs that require three colors~\cite{luzar2023locally}.

Thomassen, Wu, and Zhang proved in 2016 that a connected bipartite graph without isolated edges has the $\{1,2\}$-property if and only if it is not an odd multi-cactus~\cite{thomassen20163}.
Their argument gives the same characterization for $\{a,b\}$-edge-weightings when $a<b$, $a$ is odd, and $b$ is even.
Lyngsie proved the corresponding characterization for connected bridgeless bipartite graphs when $a=0$ and $b=1$~\cite{lyngsie2018neighbour}.
Bensmail et al.\ proved the corresponding characterization for $2$-connected bipartite graphs when $a$ is odd and $b=a+2$~\cite{bensmail2022b}.

More generally, let $\chi_\Sigma(G)$ be the smallest positive integer $k$ for which $G$ has a proper $\{1,\dots,k\}$-edge-weighting; if no such integer exists, set $\chi_\Sigma(G)=\infty$.
Early work, including the paper that introduced the $1$-$2$-$3$ conjecture, compared $\chi_\Sigma(G)$ with the chromatic number $\chi(G)$.
Karo\'{n}ski, \L{}uczak, and Thomason proved that if $(A,+)$ is a finite abelian group of odd order and $G$ has no isolated edges and is $|A|$-colorable, then an edge-weighting by elements of $A$ induces a proper vertex coloring~\cite{karonski2004edge}.
If $G$ is $2$-connected and $\chi(G)\geq 3$, then $\chi_\Sigma(G)\leq\chi(G)$~\cite{seamone20121}.
For every integer $k\geq 3$, several results give the bound $\chi_\Sigma(G)\leq k$ under coloring hypotheses.
If $k$ is odd, then the bound holds for every $k$-colorable graph $G$ without isolated edges~\cite{karonski2004edge}.
If $k\equiv 0\pmod 4$, it also holds for every $k$-colorable graph $G$ without isolated edges~\cite{duan2012factors}.
If $k\equiv 2\pmod 4$, it holds for every $k$-colorable, $2$-connected graph $G$ with minimum degree at least $k-1$~\cite{lu2009vertex}.

Addario-Berry et al.\ established the first general upper bound, $\chi_\Sigma(G)\leq 30$~\cite{addario2007vertex}.
Their method uses the degree-constrained subgraph problem.
Addario-Berry, Dalal, and Reed improved the bound to $16$~\cite{addario2005degree}, and Wang and Yu subsequently improved it to $13$~\cite{wang2008vertex}.
Kalkowski, Karo\'{n}ski, and Pfender proved that $\chi_\Sigma(G)\leq 5$; equivalently, every graph without isolated edges has the $\{1,2,3,4,5\}$-property~\cite{kalkowski2010vertex}.
Zhong proved the $1$-$2$-$3$ conjecture for sufficiently large dense graphs.
More precisely, there exists a constant $n'$ such that every graph $G=(V,E)$ with at least $n'$ vertices has the $\{1,2,3\}$-property whenever every vertex has degree greater than $0.99985|V|$~\cite{zhong20191}.
Addario-Berry, Dalal, and Reed proved that an Erd\H{o}s--R\'{e}nyi random graph has the $\{1,2\}$-property asymptotically almost surely~\cite{addario2005degree}.
For regular graphs, Przyby{\l}o proved that every $d$-regular graph with $d\geq 2$ has the $\{1,2,3,4\}$-property and that the $1$-$2$-$3$ conjecture holds when $d\geq 10^8$~\cite{przybylo20211}.
Keusch proved in 2023 that $\chi_\Sigma(G)\leq 4$~\cite{keusch2023vertex} and later confirmed the $1$-$2$-$3$ conjecture~\cite{keusch123proof}.

Related problems assign weights to vertices or use products instead of sums.
In the vertex-weighting problem, each vertex is labeled by the sum of the weights of its neighbors.
Ahadi et al.\ proved that deciding whether a graph has a proper vertex-weighting from $\{1,\dots,k\}$ is NP-complete for every $k\geq 2$, even for $3$-colorable planar graphs when $k=2$~\cite{ahadi2012computation}.
This vertex-weighting problem is also NP-complete for cubic graphs when $k=2$~\cite{dehghan2018complexity}.

In product variants, each vertex is labeled by the product of its incident edge weights or of the weights of its neighbors.
Dehghan, Sadeghi, and Ahadi proved that deciding whether a graph has a proper edge-weighting by product from $\{1,2\}$ is NP-complete both for planar $3$-colorable graphs and for cubic graphs~\cite{dehghan2013algorithmic}.
They also proved that deciding whether a proper $\{1,2\}$-vertex-weighting by product exists is NP-complete, even for planar $3$-colorable graphs~\cite{dehghan2013algorithmic}.
For every $k\geq 3$, the vertex-weighting-by-product problem with weights from $\{1,\dots,k\}$ is NP-complete on general graphs~\cite{dehghan2013algorithmic}.
Bensmail et al.\ confirmed the product analogue of the $1$-$2$-$3$ conjecture by proving that every graph without isolated edges has a proper $\{1,2,3\}$-edge-weighting by product~\cite{bensmail2023proof}.

\section{\texorpdfstring{$\{a,b\}$}{\{a,b\}}-edge-weightings in planar graphs}\label{sec:abWeightingNPC}

This section studies the decision and optimization problems for $\{a,b\}$-edge-weightings on planar graphs.
We first prove that deciding whether a proper $\{a,b\}$-edge-weighting exists is NP-complete even for simple cubic planar graphs.
We then give an exact subexponential algorithm on connected loopless planar multigraphs and a matching lower bound under the Exponential Time Hypothesis.
Finally, we give a deterministic EPTAS for maximizing the number of proper edges on connected loopless planar multigraphs.

\subsection{NP-completeness for simple cubic planar graphs}

On regular graphs, affine changes of the two allowed weights preserve properness, and the $\{0,1\}$ case is equivalent to locally irregular $2$-edge-coloring~\cite{baudon2015decomposing}.

\begin{proposition}\label{prop:regular-weight-normalization}
  Let $X$ be an $r$-regular graph, and let $a$ and $b$ be distinct integers.
  Under the bijection from $\{0,1\}$-edge-weightings $x$ to $\{a,b\}$-edge-weightings $w$ defined by $w(e)=a+(b-a)x(e)$, an edge is proper under $x$ if and only if it is proper under $w$.
  Consequently, the two weightings have the same set, and hence the same number, of proper edges.
  Moreover, viewing the two colors as $0$ and $1$, a $2$-edge-coloring of $X$ is locally irregular if and only if the resulting $\{0,1\}$-edge-weighting is proper.
  In particular, $X$ has the $\{a,b\}$-property if and only if it admits a locally irregular $2$-edge-coloring.
\end{proposition}
\begin{proof}
  Write $z_x$ and $z_w$ for the labelings induced by $x$ and $w$, respectively.
  For every vertex $v\in V(X)$, we have $z_w(v)=ra+(b-a)z_x(v)$.
  Hence $z_w(u)-z_w(v)=(b-a)\bigl(z_x(u)-z_x(v)\bigr)$ for every edge $uv\in E(X)$.
  Since $a\neq b$, the two endpoint labels differ under $w$ if and only if they differ under $x$.

  For $c\in\{0,1\}$, let $d_c(v)$ be the degree of $v$ in the subgraph induced by the edges of color $c$.
  Then $d_1(v)=z_x(v)$ and $d_0(v)=r-z_x(v)$.
  For an edge $uv$ of color $1$, local irregularity of the color-$1$ subgraph is the condition $d_1(u)\neq d_1(v)$, while for an edge of color $0$ it is the condition $d_0(u)\neq d_0(v)$.
  Since $d_0(u)\neq d_0(v)$ if and only if $d_1(u)\neq d_1(v)$, the coloring is locally irregular if and only if $z_x(u)\neq z_x(v)$ for every edge $uv$.
\end{proof}

\begin{theorem}\label{thm:planar-ab-weighting}
  For every fixed pair of distinct integers $a$ and $b$, it is NP-complete to decide whether a given simple cubic planar graph has a proper $\{a,b\}$-edge-weighting.
\end{theorem}
\begin{proof}
  The problem is in NP because a proposed edge-weighting can be checked in polynomial time.
  For NP-hardness, we reduce from the following restriction of \textsc{Monotone 2-in-4-SAT}.
  An instance is a CNF formula $\Phi$ in which every clause contains four unnegated variable occurrences and every variable occurs exactly three times in the whole formula.
  A variable may occur more than once in the same clause.
  The \emph{formula graph} of $\Phi$ is the bipartite multigraph with one vertex for each variable and each clause and one edge for each variable occurrence, joining the corresponding variable and clause vertices.
  Only instances with a planar formula graph are allowed.
  The question is whether there is a truth assignment under which exactly two of the four occurrences in every clause evaluate to true.
  We call such an assignment a $2$-in-$4$ assignment.
  This restricted problem is NP-complete~\cite[Theorem~21]{scheffler2022distance}.
  Fix a planar embedding of the formula graph for the reduction.

  By \cref{prop:regular-weight-normalization}, it is enough to prove NP-hardness for $\{0,1\}$-edge-weightings.
  All edge weights in the remainder of the proof belong to $\{0,1\}$.
  The reduction replaces each variable by a gadget with three edges joining it to the rest of the construction; these three edges have the same weight in every proper weighting.
  Each clause is replaced by a gadget with four edges joining it to the rest of the construction; its internal edges can be weighted properly exactly when two of these four edges have weight $1$.
  We establish these properties before assembling the graph.

  A \emph{rhombus gadget} has vertices $v_0,\dots,v_5$ and edge set $\{v_0v_1,v_1v_2,v_2v_3,v_3v_5,v_3v_4,v_4v_1,v_2v_4\}$.
  The edges $v_0v_1$ and $v_3v_5$ are called terminal edges, while the other five edges are called nonterminal edges.
  \Cref{fig:rhombus-gadget} shows the gadget.

  \begin{figure}[H]
    \centering
    \begin{tikzpicture}[x=\figureedgelength,y=\figureedgelength]
      \node[figure vertex,label=left:$v_0$] (v0) at ({-1-sqrt(3)/2},0) {};
      \node[figure vertex,label=above:$v_1$] (v1) at ({-sqrt(3)/2},0) {};
      \node[figure vertex,label=above:$v_2$] (v2) at (0,.5) {};
      \node[figure vertex,label=above:$v_3$] (v3) at ({sqrt(3)/2},0) {};
      \node[figure vertex,label=below:$v_4$] (v4) at (0,-.5) {};
      \node[figure vertex,label=right:$v_5$] (v5) at ({1+sqrt(3)/2},0) {};
      \draw[figure edge] (v0)--(v1)--(v2)--(v3)--(v5);
      \draw[figure edge] (v1)--(v4)--(v3);
      \draw[figure edge] (v2)--(v4);
    \end{tikzpicture}
    \caption{The rhombus gadget.}
    \label{fig:rhombus-gadget}
  \end{figure}
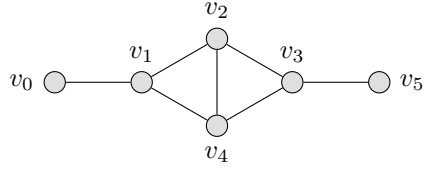

  \begin{claim}\label{claim:rhombus-gadget}
    Suppose that the five nonterminal edges of a rhombus gadget are proper.
    Then its two terminal edges have equal weights.
    If their common weight is $s\in\{0,1\}$, then the labels of $v_1$ and $v_3$ are $1+s$.
    Conversely, if both terminal edges have weight $s$, then the nonterminal edges can be weighted so that all five of them are proper.
  \end{claim}
  \begin{proof}
    For each fixed choice of the two terminal-edge weights, enumerate the $2^5$ choices of the nonterminal edge weights.
    The following table lists precisely the choices for which all five nonterminal edges are proper, together with the resulting labels of the four internal vertices.
    \[
      \begin{array}{@{}cc|ccccc|cccc@{}}
        \multicolumn{2}{@{}c|}{\text{terminal}}
        & \multicolumn{5}{c|}{\text{nonterminal}}
        & \multicolumn{4}{c@{}}{\text{labels}}\\
        v_0v_1 & v_3v_5
        & v_1v_2 & v_2v_3 & v_3v_4 & v_4v_1 & v_2v_4
        & v_1 & v_2 & v_3 & v_4\\
        \hline
        0 & 0 & 0 & 0 & 1 & 1 & 0 & 1 & 0 & 1 & 2\\
        0 & 0 & 1 & 1 & 0 & 0 & 0 & 1 & 2 & 1 & 0\\
        1 & 1 & 0 & 0 & 1 & 1 & 1 & 2 & 1 & 2 & 3\\
        1 & 1 & 1 & 1 & 0 & 0 & 1 & 2 & 3 & 2 & 1
      \end{array}
    \]
    The absence of rows with terminal-edge weights $(0,1)$ or $(1,0)$ proves that the terminal-edge weights must be equal.
    The labels of $v_1$ and $v_3$ and the converse assertion follow from the four displayed rows.
  \end{proof}

  For $k\in\{2,3\}$, a \emph{$k$-port synchronizer} is obtained from $k$ rhombus gadgets as follows.
  Take junction vertices $j_1,\dots,j_k$, indexed cyclically.
  For each $i\in\{1,\dots,k\}$, identify the two degree-one endpoints of the $i$-th rhombus gadget with $j_i$ and $j_{i+1}$, where $j_{k+1}=j_1$.
  Add one pendant edge at every junction vertex; these $k$ edges are the port edges.
  Thus every junction vertex is incident with two terminal edges and one port edge.
  The synchronizer can be embedded in a disk with the degree-one endpoints of its port edges on the boundary.
  \Cref{fig:cubic-port-synchronizers} shows the $2$- and $3$-port synchronizers.

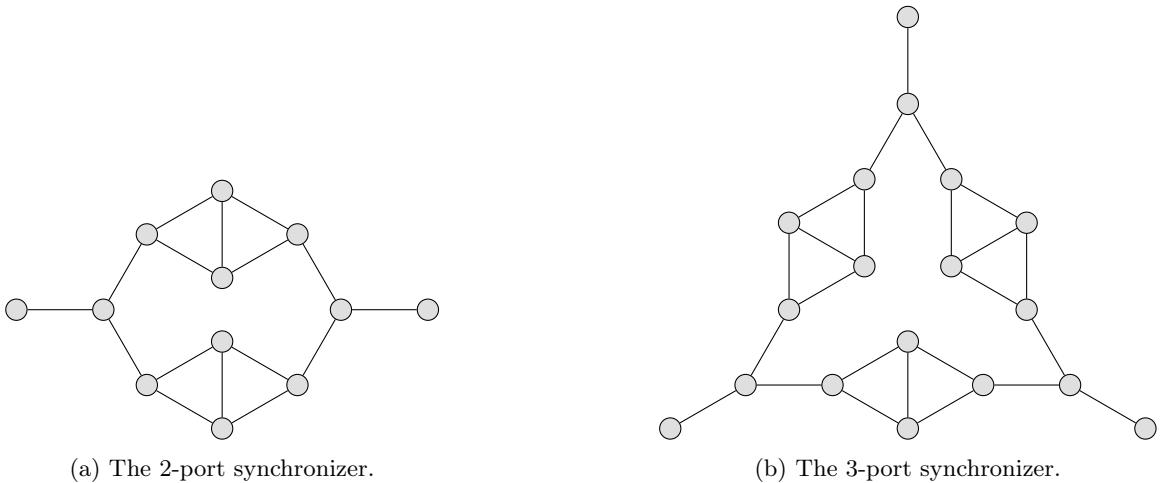
\begin{figure}[H]
  \centering
  \begin{subfigure}[b]{0.49\textwidth}
    \centering
    \begin{tikzpicture}[x=\figureedgelength,y=\figureedgelength]
      \node[figure vertex] (j1) at ({-(sqrt(3)+1)/2},0) {};
      \node[figure vertex] (j2) at ({(sqrt(3)+1)/2},0) {};
      \node[figure vertex] (p1) at ({-(sqrt(3)+3)/2},0) {};
      \node[figure vertex] (p2) at ({(sqrt(3)+3)/2},0) {};

      \node[figure vertex] (u1) at ({-sqrt(3)/2},{sqrt(3)/2}) {};
      \node[figure vertex] (u2) at (0,{(sqrt(3)+1)/2}) {};
      \node[figure vertex] (u3) at ({sqrt(3)/2},{sqrt(3)/2}) {};
      \node[figure vertex] (u4) at (0,{(sqrt(3)-1)/2}) {};

      \node[figure vertex] (l1) at ({-sqrt(3)/2},{-sqrt(3)/2}) {};
      \node[figure vertex] (l2) at (0,{-(sqrt(3)-1)/2}) {};
      \node[figure vertex] (l3) at ({sqrt(3)/2},{-sqrt(3)/2}) {};
      \node[figure vertex] (l4) at (0,{-(sqrt(3)+1)/2}) {};

      \draw[figure edge] (p1)--(j1)--(u1)--(u2)--(u3)--(j2)--(p2);
      \draw[figure edge] (u1)--(u4)--(u3);
      \draw[figure edge] (u2)--(u4);
      \draw[figure edge] (j1)--(l1)--(l2)--(l3)--(j2);
      \draw[figure edge] (l1)--(l4)--(l3);
      \draw[figure edge] (l2)--(l4);
    \end{tikzpicture}
    \caption{The $2$-port synchronizer.}
  \end{subfigure}
  \hfill
  \begin{subfigure}[b]{0.49\textwidth}
    \centering
    \begin{tikzpicture}[x=\figureedgelength,y=\figureedgelength]
      \node[figure vertex] (j1) at (0,{(2+sqrt(3))/sqrt(3)}) {};
      \node[figure vertex] (j2) at ({-(2+sqrt(3))/2},{-(2+sqrt(3))/(2*sqrt(3))}) {};
      \node[figure vertex] (j3) at ({(2+sqrt(3))/2},{-(2+sqrt(3))/(2*sqrt(3))}) {};
      \node[figure vertex] (p1) at (0,{(2+sqrt(3))/sqrt(3)+1}) {};
      \node[figure vertex] (p2) at ({-(2+sqrt(3))/2-sqrt(3)/2},{-(2+sqrt(3))/(2*sqrt(3))-.5}) {};
      \node[figure vertex] (p3) at ({(2+sqrt(3))/2+sqrt(3)/2},{-(2+sqrt(3))/(2*sqrt(3))-.5}) {};

      \drawunitrhombus{r12}{j1}{j2}
      \drawunitrhombus{r23}{j2}{j3}
      \drawunitrhombus{r31}{j3}{j1}
      \draw[figure edge] (j1)--(p1);
      \draw[figure edge] (j2)--(p2);
      \draw[figure edge] (j3)--(p3);
    \end{tikzpicture}
    \caption{The $3$-port synchronizer.}
  \end{subfigure}

  \caption{The $2$- and $3$-port synchronizers.
  The pendant edges are the port edges.}\label{fig:cubic-port-synchronizers}
\end{figure}

  \begin{claim}\label{claim:port-synchronizer}
    Let $k\in\{2,3\}$.
    If every edge lying in the rhombus gadgets of a $k$-port synchronizer is proper, then all its port edges have the same weight $s$.
    Conversely, if all port edges have weight $s$, then the remaining edges can be weighted so that every edge lying in the rhombus gadgets is proper.
    In either case, every junction vertex has label $3s$.
  \end{claim}
  \begin{proof}
    By \cref{claim:rhombus-gadget}, the two terminal edges of each rhombus gadget have equal weights.
    Suppose that the two terminal edges meeting at a junction vertex have different weights.
    Their other endpoints then have labels $1$ and $2$, while the junction vertex has label $1$ plus the weight of its port edge and hence also has label $1$ or $2$.
    One of the two terminal edges is therefore improper, a contradiction.
    Because the rhombus gadgets form a cycle, these equalities imply that all terminal edges have the same weight $s$.

    If a port edge had weight $1-s$, then its junction vertex would have label $2s+(1-s)=1+s$, equal to the labels of its two neighbors in the rhombus gadgets.
    The two terminal edges at that junction vertex would be improper.
    Thus every port edge has weight $s$.

    Conversely, assign weight $s$ to every terminal edge and apply \cref{claim:rhombus-gadget} to each rhombus gadget.
    Every junction vertex then has label $3s$, while its two neighbors in the rhombus gadgets have label $1+s$.
    Hence the terminal edges are proper, and \cref{claim:rhombus-gadget} guarantees that the nonterminal edges are proper as well.
  \end{proof}

  We next construct a graph $H$ whose four degree-two vertices will impose a relation on the weights of four attached edges.
  Its vertex set is $\{u_0,u_1,u_2,u_3,p_1,p_2,p_3,p_4\}$.
  Its edge set is $\{u_0u_1,\allowbreak u_0u_3,\allowbreak u_0p_1,\allowbreak u_1u_2,\allowbreak u_1u_3,\allowbreak u_2p_2,\allowbreak u_2p_4,\allowbreak u_3p_1,\allowbreak p_2p_3,\allowbreak p_3p_4\}$.
  The vertices $u_0,u_1,u_2,u_3$ have degree $3$, while $p_1,p_2,p_3,p_4$ have degree $2$.
  The graph has a planar embedding in which all four vertices $p_i$ lie on the outer face, as shown in \cref{fig:planar-core}.

  \begin{figure}[H]
    \centering
    \begin{tikzpicture}[x=\figureedgelength,y=\figureedgelength]
      \node[figure vertex,label=left:$p_1$] (p1) at (0,0) {};
      \node[figure vertex,label=above:$u_0$] (u0) at ({sqrt(3)/2},.5) {};
      \node[figure vertex,label=below:$u_3$] (u3) at ({sqrt(3)/2},-.5) {};
      \node[figure vertex,label=above:$u_1$] (u1) at ({sqrt(3)},0) {};
      \node[figure vertex,label=above:$u_2$] (u2) at ({sqrt(3)+1},0) {};
      \node[figure vertex,label=above:$p_2$] (p2) at ({1+3*sqrt(3)/2},.5) {};
      \node[figure vertex,label=right:$p_3$] (p3) at ({1+2*sqrt(3)},0) {};
      \node[figure vertex,label=below:$p_4$] (p4) at ({1+3*sqrt(3)/2},-.5) {};
      \draw[figure edge] (u0)--(u1)--(u3)--(p1)--(u0);
      \draw[figure edge] (u0)--(u3);
      \draw[figure edge] (u1)--(u2);
      \draw[figure edge] (u2)--(p2)--(p3)--(p4)--(u2);
    \end{tikzpicture}
    \caption{The graph $H$.}
    \label{fig:planar-core}
  \end{figure}
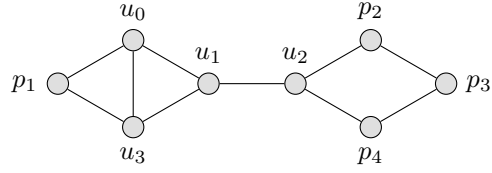

  Attach one edge at each $p_i$.
  Let its weight be $s_i\in\{0,1\}$, and suppose that its endpoint outside $H$ has label $3s_i$.

  \begin{claim}\label{claim:planar-core}
    The edges of $H$ and the four attached edges can all be proper if and only if
    \begin{equation}\label{eq:planar-core-relation}
      s_1+s_3=s_2+s_4.
    \end{equation}
  \end{claim}
  \begin{proof}
    Put $y_{12}=w(u_1u_2)$.
    First consider the subgraph induced by $p_1,u_0,u_1,u_3$.
    For fixed values of $s_1$ and $y_{12}$, enumerate the weights of its five edges in the order $u_0u_1,u_0u_3,u_0p_1,u_1u_3,u_3p_1$.
    The vertex labels in the table include the contributions of the attached edge at $p_1$ and of the edge $u_1u_2$.
    The table lists precisely the choices for which the five displayed edges and the attached edge at $p_1$ are proper; the edge $u_1u_2$ will be checked after the other side of $H$ has been weighted.
    \[
      \begin{array}{@{}w{c}{2.2em}w{c}{2.2em}|ccccc|cccc@{}}
        \multicolumn{2}{@{}c|}{\text{fixed weights}}
        & \multicolumn{5}{c|}{\text{edge weights}}
        & \multicolumn{4}{c@{}}{\text{vertex labels}}\\
        s_1 & y_{12}
        & u_0u_1 & u_0u_3 & u_0p_1 & u_1u_3 & u_3p_1
        & p_1 & u_0 & u_1 & u_3\\
        \hline
        0 & 0 & 0 & 0 & 0 & 1 & 1 & 1 & 0 & 1 & 2\\
        0 & 0 & 1 & 0 & 1 & 0 & 0 & 1 & 2 & 1 & 0\\
        1 & 1 & 0 & 1 & 0 & 1 & 1 & 2 & 1 & 2 & 3\\
        1 & 1 & 1 & 1 & 1 & 0 & 0 & 2 & 3 & 2 & 1
      \end{array}
    \]
    The table shows that the six specified edges can all be proper only when $y_{12}=s_1$, and then $z(u_1)=1+s_1$.

    We next consider the cycle $u_2p_2p_3p_4u_2$.
    Using $y_{12}=s_1$, enumerate its four edge weights in the cyclic order $u_2p_2,p_2p_3,p_3p_4,p_4u_2$.
    The following table lists precisely the weight assignments to the cycle and attached edges for which the four cycle edges and the attached edges at $p_2,p_3,p_4$ are proper.
    The displayed value of $z(u_2)$ includes the contribution $y_{12}=s_1$, and the last column checks the edge $u_1u_2$.
    \[
      \begin{array}{@{}w{c}{1.8em}w{c}{1.8em}w{c}{1.8em}w{c}{1.8em}|cccc|cc@{}}
        \multicolumn{4}{@{}c|}{\text{attached-edge weights}}
        & \multicolumn{4}{c|}{\text{cycle-edge weights}}
        & \multicolumn{2}{c@{}}{\text{check at }u_2}\\
        s_1 & s_2 & s_3 & s_4
        & u_2p_2 & p_2p_3 & p_3p_4 & p_4u_2
        & z(u_2) & z(u_2)\neq 1+s_1\\
        \hline
        0 & 0 & 0 & 0 & 0 & 1 & 1 & 0 & 0 & \text{yes}\\
        0 & 0 & 1 & 1 & 0 & 1 & 0 & 0 & 0 & \text{yes}\\
        0 & 1 & 0 & 1 & 0 & 1 & 0 & 1 & 1 & \text{no}\\
        0 & 1 & 0 & 1 & 1 & 0 & 1 & 0 & 1 & \text{no}\\
        0 & 1 & 1 & 0 & 0 & 0 & 1 & 0 & 0 & \text{yes}\\
        1 & 0 & 0 & 1 & 1 & 1 & 0 & 1 & 3 & \text{yes}\\
        1 & 0 & 1 & 0 & 0 & 1 & 0 & 1 & 2 & \text{no}\\
        1 & 0 & 1 & 0 & 1 & 0 & 1 & 0 & 2 & \text{no}\\
        1 & 1 & 0 & 0 & 1 & 0 & 1 & 1 & 3 & \text{yes}\\
        1 & 1 & 1 & 1 & 1 & 0 & 0 & 1 & 3 & \text{yes}
      \end{array}
    \]
    For the eight weight patterns on the attached edges not displayed in the table, no assignment of the four cycle edges makes all the cycle edges and attached edges proper.
    Both assignments to the cycle edges for each of the patterns $(0,1,0,1)$ and $(1,0,1,0)$ fail on $u_1u_2$.
    The weight patterns on the attached edges that admit such an assignment are therefore
    \[
      (0,0,0,0),\ (0,0,1,1),\ (0,1,1,0),\ (1,0,0,1),\ (1,1,0,0),\ (1,1,1,1),
    \]
    which are precisely the solutions of \cref{eq:planar-core-relation}.
  \end{proof}

  The next gadget relates an input edge of weight $t$ to an output edge of weight $1-t$.
  An \emph{inverting gadget} is constructed from a $3$-port synchronizer $T$ by identifying the degree-one endpoints of two port edges with one new vertex $c$.
  The third port edge is the output edge, and $c$ is the input vertex.
  In the final construction, one additional input edge is incident with $c$.
  Suppose that the input edge has weight $t$ and that its other endpoint has label $3t$.
  If every edge of the inverting gadget except its output edge is proper, then \cref{claim:port-synchronizer} implies that the three port edges of $T$ have a common weight $s$ and that their junction vertices have label $3s$.
  The vertex $c$ then has label $t+2s$, while its three neighbors have labels $3t,3s,3s$.
  Since $s,t\in\{0,1\}$, all three edges incident with $c$ are proper if and only if $s=1-t$.
  Conversely, assign common weight $1-t$ to the port edges of $T$ and use \cref{claim:port-synchronizer} to complete the synchronizer; then every edge of the inverting gadget except its output edge is proper.
  In this completion, the output edge has weight $1-t$, and its endpoint inside $T$ has label $3(1-t)$.
  The gadget can be embedded in a disk with its input vertex and the degree-one endpoint of its output edge on the boundary.

  The clause gadget is obtained from $H$ and two disjoint inverting gadgets.
  Identify the degree-one endpoint of the output edge of the first inverting gadget with $p_2$, and identify the corresponding endpoint of the second inverting gadget with $p_4$.
  The four input vertices of the clause gadget are $p_1$, the input vertex of the first inverting gadget, $p_3$, and the input vertex of the second inverting gadget, in this cyclic order.
  When the clause gadget is used in the reduction, one input edge is incident with each input vertex.

  \begin{claim}\label{claim:planar-clause-gadget}
    Let the four input-edge weights be $t_1,t_2,t_3,t_4\in\{0,1\}$ in the cyclic order just specified, and suppose that the endpoint outside the clause gadget of the $i$-th input edge has label $3t_i$.
    The remaining edges can be weighted so that every edge of the clause gadget and all four input edges are proper if and only if $t_1+t_2+t_3+t_4=2$.
  \end{claim}
  \begin{proof}
    By the preceding analysis of the inverting gadget, the two output edges must have weights $1-t_2$ and $1-t_4$, and their endpoints inside the synchronizers then have labels $3(1-t_2)$ and $3(1-t_4)$.
    Conversely, with these output-edge weights, all remaining edges of the two inverting gadgets, together with their input edges, can be made proper.
    The four edges attached to $p_1,p_2,p_3,p_4$ therefore have weights $t_1,1-t_2,t_3,1-t_4$, and in each case the endpoint outside $H$ has label three times the weight of that edge.
    By \cref{claim:planar-core}, all edges of $H$ and the four attached edges can be proper if and only if $t_1+t_3=(1-t_2)+(1-t_4)$, which is equivalent to $t_1+t_2+t_3+t_4=2$.
  \end{proof}

  The clause gadget is simple and can be embedded in a disk with its four input vertices on the boundary.
  To obtain such an embedding, start with the embedding in \cref{fig:planar-core} and place the two inverting gadgets in disjoint neighborhoods of $p_2$ and $p_4$, choosing their embeddings so that their input vertices lie on the boundary of the disk.
  Since the condition $t_1+t_2+t_3+t_4=2$ is symmetric in the four inputs, any cyclic correspondence between the input vertices and the four clause occurrences may be used.

  We construct $G_\Phi$ by replacing each variable vertex with a $3$-port synchronizer and each clause vertex with a clause gadget, while preserving the fixed planar embedding of the formula graph.
  Choose pairwise disjoint disks around its variable and clause vertices.
  For each variable $x$, place the rhombus gadgets and junction vertices of a $3$-port synchronizer $S_x$ in the variable disk, and associate its three port edges with the three occurrences of $x$ in their cyclic order.
  In each clause disk, place a clause gadget so that its four input vertices correspond, in cyclic order, to the four incident edges of the formula graph.
  For every variable occurrence, draw the corresponding port edge of the variable synchronizer along the associated edge of the formula graph and identify its degree-one endpoint with the corresponding input vertex of the clause gadget.
  This edge is then both a port edge of the variable synchronizer and an input edge of the clause gadget.
  Because each variable synchronizer and each clause gadget lies in its assigned disk and every shared port/input edge follows the corresponding embedded edge of the formula graph, the resulting graph $G_\Phi$ is planar.

  Every vertex of $G_\Phi$ has degree $3$.
  The internal vertices of the rhombus gadgets and the vertices $u_i$ have degree $3$ by construction.
  Every junction vertex is incident with two terminal edges and one port edge, while every vertex $p_i$ or $c$ is incident with two edges inside its gadget and one input or output edge.
  The graph is simple because each variable synchronizer and each clause gadget is simple and distinct occurrences use distinct junction vertices and distinct input vertices of their clause gadget.
  In particular, repeated occurrences of one variable in one clause do not create parallel edges in $G_\Phi$.
  Since each variable synchronizer and each clause gadget has constant size, the construction has linear size and can be carried out in polynomial time.

  Suppose first that $\Phi$ has a $2$-in-$4$ assignment.
  Give all three port edges of a variable synchronizer $S_x$ weight $1$ when $x$ is true and weight $0$ when $x$ is false.
  By \cref{claim:port-synchronizer}, the remaining edges of the rhombus gadgets in $S_x$ can be weighted so that they are proper, and every junction vertex then has label three times the common weight of the port edges.
  For every clause, the endpoint outside the clause gadget of each input edge is therefore a junction vertex of a variable synchronizer whose label is three times the input-edge weight.
  Exactly two input edges of the clause have weight $1$, so \cref{claim:planar-clause-gadget} gives weights for all remaining edges that make every edge of the clause gadget, including its input edges, proper.
  Repeating this for every variable and clause yields a proper $\{0,1\}$-edge-weighting of $G_\Phi$.

  Conversely, suppose that $G_\Phi$ has a proper $\{0,1\}$-edge-weighting.
  Every edge lying in the rhombus gadgets of a variable synchronizer is proper, so \cref{claim:port-synchronizer} shows that its three port edges have the same weight.
  Assign the corresponding variable the value true when this weight is $1$ and false when it is $0$.
  Each junction vertex of a variable synchronizer has label three times the common weight of the port edges.
  Hence every input edge of a clause gadget has weight $t_i$ and has an endpoint outside the clause gadget with label $3t_i$.
  Since all edges of the clause gadget and its input edges are proper, \cref{claim:planar-clause-gadget} gives $t_1+t_2+t_3+t_4=2$.
  Exactly two occurrences in every clause are therefore assigned true, so the resulting assignment is a $2$-in-$4$ assignment of $\Phi$.

  We have constructed, in polynomial time, a simple cubic planar graph $G_\Phi$ that has the $\{0,1\}$-property if and only if $\Phi$ has a $2$-in-$4$ assignment.
  Thus the problem is NP-hard for the weight set $\{0,1\}$.
  \Cref{prop:regular-weight-normalization} transfers this result to every fixed pair of distinct integers $a$ and $b$.
  Together with membership in NP, this proves the theorem.
\end{proof}

\begin{corollary}\label{cor:liec-planar-cubic-hardness}
  It is NP-complete to decide whether a given simple cubic planar graph admits a locally irregular $2$-edge-coloring.
\end{corollary}
\begin{proof}
  Membership in NP follows because local irregularity of the two color classes can be checked in polynomial time.
  By \cref{prop:regular-weight-normalization}, a cubic graph admits a locally irregular $2$-edge-coloring if and only if it has a proper $\{0,1\}$-edge-weighting.
  The NP-hardness therefore follows from \cref{thm:planar-ab-weighting} with $a=0$ and $b=1$.
\end{proof}

\subsection{Exact subexponential algorithm with a matching lower bound}\label{sec:planar-exact}

For the upper bound, we allow parallel edges but exclude loops.
Parallel edges are treated as distinct and represented explicitly.
The connected graph with no edges has one vertex and the empty weighting is proper, so we assume that $m\geq 1$.

\begin{theorem}\label{thm:planar-exact}
  Fix distinct integers $a$ and $b$.
  Given a connected loopless planar multigraph $G$ with $m\geq 1$ edges, one can either find a proper $\{a,b\}$-edge-weighting of $G$ or conclude that none exists in deterministic time $2^{O(\sqrt m)}$.
\end{theorem}
\begin{proof}
  Let $n=|V(G)|$, and let $H$ be the underlying simple graph of $G$.
  A tree decomposition is a tree whose nodes carry bags $X\subseteq V(H)$ such that every vertex appears in a bag, the bags containing any fixed vertex induce a connected subtree, and every edge has both endpoints in some bag.
  Every tree decomposition of $H$ is also a tree decomposition of $G$, since $H$ and $G$ have the same pairs of adjacent vertices.

  Write every edge weight as $w(e)=a+(b-a)x(e)$ with $x(e)\in\{0,1\}$, and for $v\in V(G)$ let $q_x(v)=\sum_{e\in E_G(v)}x(e)$.
  For $v\in V(G)$ and $r\in\{0,\dots,d_G(v)\}$, put $\lambda_v(r)=d_G(v)a+(b-a)r$.
  Then $z_w(v)=\lambda_v(q_x(v))$, so an edge $e=uv$ is proper exactly when $\lambda_u(q_x(u))\neq\lambda_v(q_x(v))$.

  We first construct a tree decomposition of $H$ using degree-dependent vertex costs.
  For $v\in V(G)$, set $c(v)=1+\left\lceil\log_2(d_G(v)+1)\right\rceil$, and for $U\subseteq V(G)$, set $\sigma(U)=\sum_{u\in U}c(u)^2$.
  We construct the decomposition recursively, starting with $U=V(G)$.
  At every recursive call, $H[U]$ is connected because $H$ is connected initially and recursive calls are made on components.
  If the current vertex set $U$ has at most one element, use $U$ as a single bag.
  Otherwise, apply the cost-and-weight separator theorem~\cite[Theorem~2]{aleksandrov2002partitioning} to the planar graph $H[U]$ with positive vertex weight $c(v)^2$, positive vertex cost $c(v)$, and $t=1/2$.
  It gives in polynomial time a separator $S\subseteq U$ such that every component $C$ of $H[U\setminus S]$ satisfies $\sigma(V(C))\leq \sigma(U)/2$ and $\sum_{v\in S}c(v)\leq 8\sqrt{\sigma(U)}$.
  Create a root bag $S$.
  For every component $C$ of $H[U\setminus S]$, recursively decompose $V(C)$, replace each bag $X$ of the resulting decomposition by $X\cup S$, and attach its root to the bag $S$.
  The root bag covers the edges within $S$, the augmented child bags cover the edges between $S$ and each component, and the remaining edges are covered recursively.
  A vertex outside $S$ occurs in one child decomposition, while a vertex in $S$ occurs in the root bag and every bag of every child decomposition, so the bags containing each vertex induce a connected subtree.
  If $|U|>1$, then $S\neq\emptyset$, since otherwise $H[U]$ would be a component of $H[U\setminus S]$ with weight $\sigma(U)>\sigma(U)/2$.
  The separator sets and base-case singletons from distinct recursive calls are pairwise disjoint, so there are $O(n)$ calls and hence $O(n)$ bags.

  Every bag is the union of the separators from an initial segment of a root-to-leaf path in the recursion tree, together with at most one base-case vertex.
  If a recursive call at depth $i$ has vertex set $U$, then $\sigma(U)\leq 2^{-i}\sigma(V(G))$, so every resulting bag $X$ satisfies
  \begin{equation}\label{eq:weighted-bag-cost}
    \sum_{v\in X}c(v)
    =
    O\left(
      \sum_{i=0}^{\infty}\sqrt{2^{-i}\sigma(V(G))}
    \right)
    =
    O\left(\sqrt{\sigma(V(G))}\right).
  \end{equation}
  Convert this decomposition into a nice tree decomposition with an empty root and empty leaves, introduce-vertex, introduce-edge, forget-vertex, and join nodes.
  This uses only duplicate bags and subsets of existing bags, so it does not increase the maximum value of $\sum_{v\in X}c(v)$.
  For every edge $e\in E(G)$, choose a bag containing its endpoints and introduce $e$ exactly once at a new introduce-edge node.

  We now run a dynamic program over the nice tree decomposition.
  For a node $t$, let $X_t$ be its bag, let $V_t$ be the set of vertices appearing in bags of its subtree, and let $E_t$ be the edges introduced in that subtree.
  Since the bags containing any fixed vertex induce a connected subtree, all bags containing a vertex in $V_t\setminus X_t$ lie in the subtree of $t$; hence every edge incident with such a vertex belongs to $E_t$.
  A state assigns a pair $(p_v,q_v)$ to every $v\in X_t$, where $0\leq p_v\leq q_v\leq d_G(v)$.
  The value $p_v$ is the number of edges of $E_t$ incident with $v$ that have weight $b$, while $q_v$ is the total number of edges incident with $v$ that will have weight $b$ in the completed weighting.
  A state is feasible if one can choose $q_v\in\{0,\dots,d_G(v)\}$ for every $v\in V_t\setminus X_t$ and weight all edges of $E_t$ so that the following conditions hold.
  For every $v\in X_t$, exactly $p_v$ edges of $E_t$ incident with $v$ have weight $b$.
  For every $v\in V_t\setminus X_t$, exactly $q_v$ incident edges have weight $b$.
  Finally, every edge $e=uv\in E_t$ satisfies $\lambda_u(q_u)\neq\lambda_v(q_v)$.

  At an empty leaf, the unique state on the empty bag is feasible.
  At an introduce-vertex node for $v$, choose $q_v\in\{0,\dots,d_G(v)\}$ and set $p_v=0$.
  At an introduce-edge node for an edge $e=uv$, discard a state unless $\lambda_u(q_u)\neq\lambda_v(q_v)$.
  For each remaining state, consider both possible weights of $e$.
  If $e$ receives weight $a$, leave $p_u$ and $p_v$ unchanged.
  If $e$ receives weight $b$, increase both partial counts by one, discarding the result if either count exceeds the corresponding $q$-value.
  At a forget-vertex node for $v$, retain a child state only if $p_v=q_v$, and then delete the pair for $v$.
  At a join node, combine two child states only when their $q_v$-values agree for every bag vertex, and set each $p_v$ to the sum of the two child $p_v$-values.
  The two child subtrees have no common vertex outside the bag, and their introduced edge sets are disjoint, so these sums are correct.
  Discard a combination if $p_v>q_v$ for some $v$.

  Induction on the decomposition shows that these transitions compute exactly the feasible states.
  At the empty root, the unique state is feasible exactly when $G$ has a proper weighting.

  It remains to bound the number of states.
  For a vertex of degree $d$, there are
  \[
    \sum_{q=0}^{d}(q+1)
    =
    \frac{(d+1)(d+2)}{2}
    \leq
    (d+1)^2
  \]
  possible pairs $(p_v,q_v)$.
  Hence a bag $X$ has at most
  \[
    \prod_{v\in X}(d_G(v)+1)^2
    =
    2^{O\left(\sum_{v\in X}c(v)\right)}
    =
    2^{O(\sqrt{\sigma(V(G))})}
  \]
  states by \cref{eq:weighted-bag-cost}.
  Since $c(v)^2=O(d_G(v)+1)$ for every $v$,
  \[
    \sigma(V(G))
    =
    O\Bigg(n+\sum_{v\in V(G)}d_G(v)\Bigg)
    =
    O(n+m).
  \]
  Connectedness and $m\geq 1$ give $n\leq m+1$, so $\sigma(V(G))=O(m)$.
  The nice decomposition has polynomially many nodes, and a naive join transition is quadratic in the number of states, so the total running time is $2^{O(\sqrt m)}$.
  For each feasible state, store the child state or states used to obtain it and, at an introduce-edge node, the chosen edge weight.
  If the root state is feasible, these records recover a proper weighting within the same time bound.
\end{proof}

Next we show that the dependence on $\sqrt m$ in the exponent is optimal under the Exponential Time Hypothesis, even for connected simple cubic planar graphs.

\begin{theorem}\label{thm:planar-eth-lower-bound}
  Fix distinct integers $a$ and $b$.
  Assuming the Exponential Time Hypothesis, there is no $2^{o(\sqrt m)}$-time algorithm for deciding whether a connected simple cubic planar graph with $m$ edges has a proper $\{a,b\}$-edge-weighting.
\end{theorem}
\begin{proof}
  By the sparsification lemma, the Exponential Time Hypothesis implies that \textsc{3-SAT} with $N$ variables and $O(N)$ clauses cannot be solved in time $2^{o(N)}$~\cite{impagliazzo2001which}.
  We first give a linear-size reduction from this sparse problem to \textsc{Monotone 2-in-4-SAT}.
  Write $R(y_1,y_2,y_3,y_4)$ for the constraint that exactly two of its four arguments are true.

  Introduce a global variable $f$, and replace every clause $\ell_1\lor\ell_2\lor\ell_3$ by $\operatorname{NAE}(\ell_1,\ell_2,\ell_3,f)$, where $\operatorname{NAE}$ means that its arguments are not all equal.
  Every satisfying assignment of the original formula extends to these constraints by setting $f=0$.
  Conversely, let an assignment satisfy all these constraints.
  If $f=0$, its restriction to the original variables satisfies the original formula.
  If $f=1$, complementing all original variables gives a satisfying assignment of the original formula.
  Replace every four-argument NAE constraint by two three-argument constraints with a fresh variable $u$, using $\operatorname{NAE}(p,q,r,s)$ if and only if $\exists u\,\bigl(\operatorname{NAE}(p,q,u)\mathbin{\wedge}\operatorname{NAE}(\neg u,r,s)\bigr)$.
  For fixed values of $p,q,r,s$, such a value of $u$ exists exactly when the four values are not all equal.

  For every variable $x$ of the resulting \textsc{NAE-3-SAT} instance, introduce variables $x^0$ and $x^1$ and add the constraint $R(x^0,x^0,x^1,x^1)$.
  Represent the literal $x$ by $x^0$ and the literal $\neg x$ by $x^1$.
  For every constraint $\operatorname{NAE}(y_1,y_2,y_3)$, introduce a fresh variable $t$ and add $R(y_1,y_2,y_3,t)$.
  The consistency constraint enforces $x^0\neq x^1$, and the constraint involving $t$ admits a value of $t$ exactly when one or two of $y_1,y_2,y_3$ are true.
  The resulting monotone $2$-in-$4$ formula $\Psi$ is equisatisfiable with the original formula and has $O(N)$ variable occurrences.
  Its formula graph is connected.
  For each original clause, the two resulting constraints are both connected to the constraint $R(u^0,u^0,u^1,u^1)$ for the variable $u$ introduced when the corresponding four-argument NAE constraint is split, and one of them contains $f^0$.

  We next track the size of the planarization used in the proof of \cref{thm:planar-ab-weighting}.
  In Scheffler's planarization construction~\cite[Lemma~20]{scheffler2022distance}, every clause is copied twice, and the formula graph, which still has $O(N)$ edges, is drawn so that every pair of edges crosses at most once.
  The drawing therefore has $O(N^2)$ crossings.
  Each crossing-elimination step adds a constant number of variables and clauses and reduces the number of crossings by one without creating new crossings.
  It follows that the resulting planar formula has size $O(N^2)$ and every variable occurs a multiple of three times.
  The variable-replacement construction~\cite[Theorem~21]{scheffler2022distance} replaces a variable with $3k$ occurrences by $O(k)$ variables and clauses while preserving planarity.
  Summing over all variables preserves the $O(N^2)$ size bound, so we obtain an equivalent planar \textsc{Monotone 2-in-4-SAT} instance $\Phi$ in which every variable occurs exactly three times.
  Duplicating clauses preserves connectedness.
  In the crossing-elimination construction~\cite[Lemma~20]{scheffler2022distance}, each crossing is replaced by a connected subgraph attached to the endpoints of the two crossing edges.
  In the variable-replacement construction~\cite[Theorem~21]{scheffler2022distance}, the subgraph replacing a variable is connected and contains all of its clause incidences.
  Hence the formula graph remains connected throughout these transformations, and the formula graph of $\Phi$ is connected.

  Apply the construction from the proof of \cref{thm:planar-ab-weighting} to $\Phi$.
  The construction is linear in the size of $\Phi$, so the resulting simple cubic planar graph $G_\Phi$ has $m_\Phi=O(N^2)$ edges.
  Each variable synchronizer and clause gadget is connected, and every edge of the formula graph is represented by an edge joining the corresponding pair, so $G_\Phi$ is connected.
  Moreover, $G_\Phi$ has a proper $\{0,1\}$-edge-weighting if and only if $\Phi$ has a $2$-in-$4$ assignment.

  A $2^{o(\sqrt m)}$-time algorithm for the $\{0,1\}$-edge-weighting problem on connected simple cubic planar graphs would decide $G_\Phi$ in time $2^{o(\sqrt{m_\Phi})}=2^{o(N)}$, and thus solve sparse \textsc{3-SAT} in time $2^{o(N)}$, contradicting the Exponential Time Hypothesis.

  This proves the lower bound for the weight set $\{0,1\}$.
  Finally, \cref{prop:regular-weight-normalization} transfers it to every fixed pair of distinct integers $a$ and $b$.
\end{proof}

\begin{corollary}\label{cor:liec-planar-cubic-subexp}
  Given a simple cubic planar graph with $n$ vertices, one can decide whether it admits a locally irregular $2$-edge-coloring in deterministic time $2^{O(\sqrt n)}$.
  Assuming the Exponential Time Hypothesis, no $2^{o(\sqrt n)}$-time algorithm exists for this problem.
\end{corollary}
\begin{proof}
  Let the connected components have $n_1,\dots,n_k$ vertices.
  Each component is cubic and has $3n_i/2$ edges, and by \cref{prop:regular-weight-normalization} it admits a locally irregular $2$-edge-coloring if and only if it has a proper $\{0,1\}$-edge-weighting.
  Applying \cref{thm:planar-exact} with $a=0$ and $b=1$ to each component takes total time at most $n2^{O(\sqrt n)}=2^{O(\sqrt n)}$.
  The lower bound follows from \cref{thm:planar-eth-lower-bound}, which already holds for connected simple cubic planar graphs.
\end{proof}

\subsection{An EPTAS for the maximization problem}\label{sec:planar-eptas}

In this subsection, we allow parallel edges but exclude loops.
Parallel edges are treated as distinct, so each copy contributes separately to the incident-edge sets, the vertex labels, and the objective value.

The problem \textsc{Max $\{a,b\}$-edge-weighting} asks for an $\{a,b\}$-edge-weighting that maximizes the number of proper edges.
For a graph $G$, let $\OPT_{\mathrm{EW}}(G)$ denote this maximum.
\Cref{thm:planar-ab-weighting} implies that computing $\OPT_{\mathrm{EW}}(G)$ is NP-hard even for simple cubic planar graphs, because such a graph has the $\{a,b\}$-property exactly when $\OPT_{\mathrm{EW}}(G)=|E(G)|$.
We nevertheless obtain a deterministic EPTAS on connected loopless planar multigraphs.
We first bound the running time in the arithmetic model, in which every arithmetic operation on integers takes unit time.

The scheme constructs one candidate weighting for each choice of a congruence class of breadth-first search (BFS) layers and a prime modulus $p$.
Deleting the chosen layers leaves components contained in a bounded number of consecutive layers.
For each remaining component, a dynamic program chooses the weights of its edges and, for every vertex, a residue representing the total weight of incident edges joining it to deleted vertices.
Consider an edge that is proper under an optimal weighting.
It is absent from the remaining graph only if one of its endpoints is deleted; otherwise, the residues modulo $p$ of its endpoint labels are equal only if $p$ divides the nonzero difference of its endpoint labels.
Averaging over the layer classes and the selected primes yields a pair for which at least a $(1-\varepsilon)$-fraction of the optimal proper edges have neither endpoint deleted and have endpoint labels with distinct residues modulo $p$.
The modular subproblem is solved by a dynamic program on a tree derived from a planar embedding.
Every state records modular information for only $O(k)$ vertices.

To formulate the dynamic program, suppose that $q$ edges incident with a vertex join it to deleted layers.
If exactly $t$ of these edges have weight $b$, then their total weight is $qa+t(b-a)$.
For $q\geq 0$ and $p\geq 2$, let $\mathcal R_p(q)\subseteq\{0,\dots,p-1\}$ be the set of residues modulo $p$ of $qa+t(b-a)$ as $t$ ranges over $\{0,\dots,q\}$.
A choice $r(v)\in\mathcal R_p(q(v))$ can therefore represent the contribution to the label of $v$ from the $q(v)$ omitted incident edges.
For a subgraph $C$, a modulus $p$, an edge-weighting $x:E(C)\rightarrow\{a,b\}$, and residues $r(v)\in\{0,\dots,p-1\}$ for $v\in V(C)$, let $\bar z_{x,r}(v)$ denote the residue modulo $p$ of $r(v)+\sum_{e\in E_C(v)}x(e)$.
The following lemma solves the resulting modular optimization problem for a component contained in at most $k-1$ consecutive BFS layers.

\begin{lemma}\label{lem:planar-modular-strip}
  Let $k\geq 2$, let $G$ be a connected loopless planar multigraph, and fix a root $s\in V(G)$.
  For $i\geq 0$, let $L_i=\{v\in V(G):\dist_G(s,v)=i\}$ be the $i$-th BFS layer rooted at $s$.
  Let $0\leq i_0\leq i_1$ satisfy $i_1-i_0\leq k-2$, and let $C$ be any connected component of $G[L_{i_0}\cup\dots\cup L_{i_1}]$ that meets both $L_{i_0}$ and $L_{i_1}$.
  For each $v\in V(C)$, let $q(v)$ be an integer with $0\leq q(v)\leq d_G(v)$.
  Finally, let $p\geq 2$ be a modulus.

  After an $O(|V(G)|+|E(G)|)$-time preprocessing step independent of $p$, one can compute an edge-weighting $x:E(C)\rightarrow\{a,b\}$ and residues $r(v)\in\mathcal R_p(q(v))$ that maximize $\bigl|\{e=vv'\in E(C):\bar z_{x,r}(v)\neq\bar z_{x,r}(v')\}\bigr|$ using $O\bigl(p^{13k+12}(|E(C)|+1)\bigr)$ arithmetic operations.
  The algorithm also returns, for each $v\in V(C)$, an integer $t(v)\in\{0,\dots,q(v)\}$ satisfying $r(v)\equiv q(v)a+t(v)(b-a)\pmod p$.
\end{lemma}

\begin{proof}
  If $C$ has no edges, then it consists of one vertex $v$.
  Set $t(v)=0$, let $r(v)$ be the residue of $q(v)a$ modulo $p$, and return the empty edge-weighting.
  Henceforth assume that $C$ has at least one edge.

  Fix a planar embedding of $G$, and use the embeddings induced by the constructions below.
  We first construct a loopless planar multigraph $J$, embedded in the plane, that contains $C$ and has a spanning tree of depth at most $k-1$.
  Suppose that $i_0>0$, and let $G_{<i_0}=G[L_0\cup\dots\cup L_{i_0-1}]$.
  The graph $G_{<i_0}$ is connected because every vertex in $G_{<i_0}$ has a shortest path to $s$ that remains in $G_{<i_0}$.
  For each $v\in C\cap L_{i_0}$, choose one edge from $v$ to $L_{i_0-1}$ on a shortest path from $v$ to $s$.
  In the embedded subgraph consisting of $C$, $G_{<i_0}$, and these chosen edges, contract a spanning tree of $G_{<i_0}$ to a single vertex $s_0$ and delete every remaining edge of $G_{<i_0}$, which becomes a loop after the contraction.
  Let $J$ be the resulting loopless planar multigraph, with the embedding inherited from the contraction.
  It contains $C$ and one chosen edge $s_0v$ for every $v\in C\cap L_{i_0}$.

  If $i_0=0$, then $C\cap L_0=\{s\}$.
  In this case, let $J$ be obtained from $C$ by adding a new vertex $s_0$ in a face incident with $s$ and drawing the edge $s_0s$ inside that face.

  For every $v\in C\cap L_i$ with $i>i_0$, choose an edge from $v$ to $L_{i-1}$ on a shortest path from $v$ to $s$.
  Its endpoint in $L_{i-1}$ belongs to $C$, because $C$ is a component of the graph induced by $L_{i_0},\dots,L_{i_1}$.
  Following the chosen edges toward lower-indexed layers eventually reaches $L_{i_0}$ and then $s_0$, and the strict decrease in layer index prevents a cycle.
  Together with the edges from $s_0$ to $C\cap L_{i_0}$, the chosen edges therefore form a spanning tree $T_0$ of $J$.
  For $v\in C\cap L_i$, the unique path in $T_0$ from $s_0$ to $v$ has length $1+i-i_0$, so $\dist_{T_0}(s_0,v)=1+i-i_0\leq k-1$.

  We next augment the embedded graph $J$ to a loopless planar multigraph $K$ whose faces are triangular.
  View the embedding on the sphere, so that no face is distinguished as the outer face.
  A facial boundary walk need not be a cycle and may visit the same vertex more than once, for example when that vertex is a cutvertex.
  Treat each visit as a separate occurrence.
  In each face, add one new vertex and join it to every vertex occurrence on the boundary walk, drawing the new edges inside the face in the cyclic order of the occurrences.
  Repeated occurrences may create parallel edges, but they cannot create a loop because the face vertex is new.
  The new edges divide the original face into triangular faces.
  For each new face vertex, choose one incident new edge.
  These selected edges together with $T_0$ form a spanning tree $T$ of $K$.
  Every new face vertex is adjacent in $T$ to a vertex at distance at most $k-1$ from $s_0$, so $T$ has depth at most $k$.
  Apart from the vertices and edges of $C$, the graph $J$ contains only $s_0$ and at most one additional edge for each vertex in $C\cap L_{i_0}$.
  Hence $|V(J)|+|E(J)|=O(|V(C)|+|E(C)|)$.
  Euler's formula bounds the number of faces of $J$, and the total number of occurrences on all facial boundary walks is $2|E(J)|$.
  Since the augmentation adds one vertex per face and one edge per boundary occurrence, $|V(K)|+|E(K)|=O(|V(C)|+|E(C)|)$.

  We define the dynamic program on a tree whose vertices correspond to the faces of $K$.
  Let $K^*$ be the geometric dual of $K$, and let $D$ be the spanning subgraph of $K^*$ consisting of the duals of the edges in $E(K)\setminus E(T)$.
  Euler's formula gives $|E(K)\setminus E(T)|=|V(K^*)|-1$.
  No edge of $D$ is a loop, because a loop in the dual corresponds to a bridge of $K$, and every bridge of $K$ belongs to the spanning tree $T$.
  To see that $D$ is acyclic, suppose that it contains a cycle.
  In the dual embedding, this cycle gives a simple closed curve that crosses a nonempty cut of $K$ only in edges outside $T$.
  This is impossible because the spanning tree $T$ contains an edge of every nonempty cut.
  Hence $D$ is acyclic, and since it is a spanning subgraph with $|V(K^*)|-1$ edges, it is a spanning tree of $K^*$.
  Because every face of $K$ is triangular, $D$ has maximum degree at most three.

  Assign each edge of $C$ to one of its incident faces in $K$, regard each vertex of $D$ as its corresponding face, and root $D$ at an arbitrary vertex $u_0$.
  At most three edges of $C$ are assigned to any vertex of $D$.
  The planar embedding, contractions, augmentation, dual construction, and edge assignment can all be performed in $O(|V(G)|+|E(G)|)$ time, and none of these operations depends on $p$.
  For $u\in V(D)$, let $D_u$ be the subtree rooted at $u$, and let $E_u$ be the set of edges of $C$ assigned to the faces corresponding to vertices of $D_u$.
  We show that only a bounded number of vertices of $C$ are incident with both an edge of $E_u$ and an edge outside $E_u$.

  Let $u\in V(D)\setminus\{u_0\}$.
  Write $g_u^*$ for the unique edge of $D$ joining $D_u$ to the rest of $D$, and let $g_u$ be the corresponding edge of $K$, with endpoints $v_{u,1}$ and $v_{u,2}$.
  Let $Q_u$ be the set of edges of $K$ with one incident face corresponding to a vertex of $D_u$ and the other corresponding to a vertex outside $D_u$.
  The only edge of $Q_u$ outside $T$ is $g_u$, because $g_u^*$ is the only dual-tree edge between the two sets of faces.
  Moreover, every vertex of $K$ has even degree in the subgraph with edge set $Q_u$.
  To see this, traverse the incident face occurrences cyclically around a vertex and mark whether the corresponding vertex of $D$ belongs to $D_u$.
  An incident edge lies in $Q_u$ exactly at a change between the two marks, and every cyclic binary sequence has an even number of changes.
  The graph obtained from $T$ by adding $g_u$ has a unique cycle.
  The subgraph with edge set $Q_u$ is nonempty, contains $g_u$, and has even degree at every vertex.
  Since $T$ together with $g_u$ contains exactly one cycle, $Q_u$ is exactly the edge set of that cycle.
  Thus $Q_u=\{g_u\}\cup E(P_T(v_{u,1},v_{u,2}))$, where $P_T(v_{u,1},v_{u,2})$ is the unique $v_{u,1}$--$v_{u,2}$ path in $T$.
  Therefore
  \[
    |Q_u|
    \leq
    1+\dist_T(s_0,v_{u,1})+\dist_T(s_0,v_{u,2})
    \leq 2k+1.
  \]

  Let $I_u$ be the set of vertices of $C$ incident with both an edge of $E_u$ and an edge of $E(C)\setminus E_u$.
  If $v\in I_u$, then among the face occurrences around $v$ there is a transition from a face corresponding to a vertex of $D_u$ to one corresponding to a vertex outside $D_u$.
  The edge between these two consecutive face occurrences belongs to $Q_u$ and is incident with $v$.
  Since $Q_u$ is the edge set of a cycle, it is incident with exactly $|Q_u|$ vertices.
  Hence $|I_u|\leq |Q_u|\leq 2k+1$.
  At the root $u_0$, we have $E_{u_0}=E(C)$ and $I_{u_0}=\emptyset$.

  Before processing the vertices of $D$, construct for every $v\in V(C)$ a lookup table indexed by $r\in\{0,\dots,p-1\}$.
  For every $t\in\{0,\dots,\min\{q(v),p-1\}\}$, let $r$ be the residue satisfying $r\equiv q(v)a+t(b-a)\pmod p$ and store one such value $t$ in the entry indexed by $r$.
  Leave every index not obtained in this way undefined.
  Since the residue depends only on $t$ modulo $p$, an entry indexed by $r$ is defined exactly when $r\in\mathcal R_p(q(v))$.
  The lookup tables require $O(p|V(C)|)$ arithmetic operations and allow membership in every set $\mathcal R_p(q(v))$ to be tested in $O(1)$ arithmetic operations.

  Process the vertices of $D$ from the leaves to the root.
  A state $S$ at $u$ assigns to each $v\in I_u$ a pair $(\bar z_v,y_v)$ with $\bar z_v,y_v\in\{0,\dots,p-1\}$.
  The coordinate $\bar z_v$ is the residue modulo $p$ of the final label at $v$, and $y_v$ is the residue of the partial sum $\sum_{e\in E_u\cap E_C(v)}x(e)$.
  For every state $S$ at $u$, the table stores the maximum value of
  \[
    \bigl|\{e=vv'\in E_u:\bar z_v\neq\bar z_{v'}\}\bigr|
  \]
  over all weight assignments $x:E_u\rightarrow\{a,b\}$ and all choices $\bar z_v\in\{0,\dots,p-1\}$ for vertices $v$ incident with an edge of $E_u$ that satisfy the following conditions.
  For every $v\in I_u$, the pair $(\bar z_v,y_v)$ equals the pair assigned to $v$ by $S$, and
  \[
    y_v\equiv\sum_{e\in E_u\cap E_C(v)}x(e)\pmod p.
  \]
  For every vertex $v$ incident with an edge of $E_u$ but not belonging to $I_u$, there exists $r\in\mathcal R_p(q(v))$ such that
  \[
    \bar z_v-\sum_{e\in E_C(v)}x(e)\equiv r\pmod p.
  \]
  Every edge of $C$ incident with such a vertex belongs to $E_u$ by the definition of $I_u$, so the sum in this congruence is determined by $x$.
  If no choices of $x$ and the values $\bar z_v$ satisfy these conditions, the state is omitted from the table.

  We prove by induction on $D_u$ that the tables have these values while describing how to compute them.
  Let $u_1,\dots,u_h$ be the children of $u$.
  The sets $E_{u_1},\dots,E_{u_h}$ together with the at most three edges of $C$ assigned directly to $u$ form a partition of $E_u$.
  Choose one state from each child table and enumerate the weights of the edges assigned directly to $u$.
  If $v\in I_{u_i}\cap I_{u_j}$ for two selected child states, require the two states to assign the same value $\bar z_v$.
  If a vertex $v$ is incident with an edge of $E_{u_i}$ and $v\notin I_{u_i}$, then every edge of $C$ incident with $v$ lies in $E_{u_i}$, so the congruence condition for $v$ has already been checked in the table of $u_i$.
  For each endpoint $v$ of an edge assigned directly to $u$, use the common value $\bar z_v$ from the selected child states if $v\in I_{u_i}$ for at least one child $u_i$, and otherwise enumerate $\bar z_v\in\{0,\dots,p-1\}$.
  Every vertex $v\in I_u$ either is an endpoint of an edge assigned directly to $u$ or belongs to $I_{u_i}$ for some child $u_i$.
  To verify this statement, choose an edge of $E_u$ incident with $v$.
  If that edge belongs to $E_{u_i}$, then an edge incident with $v$ outside $E_u$ also lies outside $E_{u_i}$, and hence $v\in I_{u_i}$.
  For every vertex $v$ such that $v\in I_{u_i}$ for at least one child $u_i$ or $v$ is an endpoint of an edge assigned directly to $u$, let $\widehat y_v$ be the sum modulo $p$ of the $y_v$-coordinates from all selected child states with $v\in I_{u_i}$ and the weights of all edges assigned directly to $u$ that are incident with $v$.
  If $v\in I_u$, assign the pair $(\bar z_v,\widehat y_v)$ to $v$ in the parent state.
  If $v\notin I_u$, require that there exists $r\in\mathcal R_p(q(v))$ such that
  \[
    \bar z_v-\widehat y_v\equiv r\pmod p.
  \]
  This condition is tested by the lookup table for $v$.
  For every parent state obtained in this way, compare its current table value with the sum of the values stored for the selected child states plus the number of edges $e=vv'$ assigned directly to $u$ for which $\bar z_v\neq\bar z_{v'}$, and keep the larger value.

  We verify that the transition computes the stated table values.
  By the induction hypothesis, for each selected child state choose weights on its edge set and values $\bar z_v$ that attain the value stored for that state.
  If a vertex $v$ is incident with edges in both $E_{u_i}$ and $E_{u_j}$, then $v\in I_{u_i}\cap I_{u_j}$, and the transition requires the same value $\bar z_v$ in both child states.
  The same value is used when $v\in I_{u_i}$ for some child $u_i$ and $v$ is incident with an edge assigned directly to $u$.
  If $v$ is incident with an edge of $E_{u_i}$ and $v\notin I_{u_i}$, then every edge of $C$ incident with $v$ lies in $E_{u_i}$, so the child table has already checked the congruence for $v$ and no other child edge set or edge assigned directly to $u$ is incident with $v$.
  For every remaining vertex $v\notin I_u$, the transition checks the congruence after all edges of $C$ incident with $v$ have been included in $E_u$.
  Thus the combined choices satisfy the conditions for the parent state, and, since the child edge sets together with the edges assigned directly to $u$ partition $E_u$, the value computed by the transition is exactly their objective value.

  Conversely, fix an edge-weighting $x:E_u\rightarrow\{a,b\}$ and values $\bar z_v$ satisfying the conditions for a state at $u$.
  For each child $u_i$, restrict $x$ to $E_{u_i}$ and, for every $v\in I_{u_i}$, use the same value $\bar z_v$ and set
  \[
    y_v\equiv\sum_{e\in E_{u_i}\cap E_C(v)}x(e)\pmod p.
  \]
  If a vertex $v$ is incident with an edge of $E_{u_i}$ but does not belong to $I_{u_i}$, then all edges of $C$ incident with $v$ lie in $E_{u_i}$, so $v\notin I_u$ and the congruence in the child state follows from the corresponding condition for the parent state.
  Hence these restrictions define valid child states, and the transition considers them together with the given weights of the edges assigned directly to $u$ and, for every endpoint $v$ of a directly assigned edge that belongs to no $I_{u_i}$, the given value $\bar z_v$.
  By the induction hypothesis, the value stored for each selected child state is at least the objective value of the corresponding restriction, so the transition obtains the parent state with value at least the objective value of the fixed choices of $x$ and $\bar z_v$.
  Together with the preceding direction, this proves the stated table invariant, with the case $h=0$ giving the induction base.

  At the root $u_0$, we have $I_{u_0}=\emptyset$ and $E_{u_0}=E(C)$.
  For every edge-weighting $x$ and values $\bar z_v$ counted by the root table, the root conditions give, for each $v\in V(C)$, an $r(v)\in\mathcal R_p(q(v))$ satisfying
  \[
    \bar z_v-\sum_{e\in E_C(v)}x(e)\equiv r(v)\pmod p.
  \]
  Hence $\bar z_{x,r}(v)=\bar z_v$ for every $v\in V(C)$, so the value stored in the root table is at most the optimum in the statement of the lemma.
  Conversely, every feasible choice of $x$ and $r$ in the modular optimization problem satisfies the root conditions when $\bar z_v=\bar z_{x,r}(v)$ for every $v\in V(C)$.
  The value stored in the root table is therefore exactly the optimum in the statement of the lemma.

  Whenever a table value is updated, store the selected child states, the weights of the edges assigned directly to the current vertex of $D$, and every value $\bar z_v$ enumerated at that transition.
  Backtracking from the root determines $x$ and the value $\bar z_v$ for every vertex of $C$ incident with an edge.
  For each such vertex $v$, let $r(v)$ be the unique residue in $\{0,\dots,p-1\}$ satisfying
  \[
    \bar z_v-\sum_{e\in E_C(v)}x(e)\equiv r(v)\pmod p.
  \]
  The lookup entry indexed by $r(v)$ is defined by the table conditions and gives a value $t(v)\in\{0,\dots,q(v)\}$ satisfying $r(v)\equiv q(v)a+t(v)(b-a)\pmod p$.

  For each $v\in I_u$, there are at most $p^2$ possible pairs $(\bar z_v,y_v)$.
  Since $|I_u|\leq 2k+1$, each table has at most $p^{4k+2}$ states.
  A vertex of $D$ has at most three children, so there are at most $p^{12k+6}$ choices of child states.
  The edges assigned directly to a vertex $u$ of $D$ have at most six endpoints, and for each such endpoint $v$ satisfying $v\notin I_{u_i}$ for every child $u_i$ there are at most $p$ choices for $\bar z_v$.
  There are at most eight weight assignments to the at most three edges assigned directly to the vertex of $D$.
  For a fixed combination, every vertex that must be examined belongs to $I_{u_i}$ for some child $u_i$ or is an endpoint of an edge assigned directly to $u$, so there are at most $3(2k+1)+6=6k+9$ such vertices.
  Each lookup-table test uses $O(1)$ arithmetic operations, so processing one combination uses $O(k)$ arithmetic operations.
  Since $k\leq p^k$ for $k\geq 2$ and $p\geq 2$, each vertex of $D$ can be processed using $O(p^{13k+12})$ arithmetic operations.
  The tree $D$ has $O(|V(C)|+|E(C)|)$ vertices.
  Since $C$ is connected and has at least one edge, $|V(C)|\leq |E(C)|+1$.
  The dynamic program and the lookup-table construction therefore use $O\bigl(p^{13k+12}(|E(C)|+1)\bigr)$ arithmetic operations after the preprocessing independent of $p$.
\end{proof}

\begin{theorem}\label{thm:planar-eptas}
  For every fixed pair of distinct integers $a$ and $b$, \textnormal{\textsc{Max $\{a,b\}$-edge-weighting}} admits a deterministic EPTAS on connected loopless planar multigraphs.
  Given $0<\varepsilon<1$ and a connected loopless planar multigraph $G$ with $m\geq 1$ edges, the algorithm returns an $\{a,b\}$-edge-weighting of $G$ with at least $(1-\varepsilon)\OPT_{\mathrm{EW}}(G)$ proper edges.
  Its running time is $2^{O((1/\varepsilon)\log(2/\varepsilon))}m^{O(1)}$.
\end{theorem}

\begin{proof}
  Fix $0<\varepsilon<1$ and a connected loopless planar multigraph $G$ with $m\geq 1$ edges.
  Since $G$ is connected, it has at most $m+1$ vertices.
  Define
  \[
    W_0=\max\{1,|a|,|b|\},\qquad
    k=\left\lceil\frac{4}{\varepsilon}\right\rceil,\qquad
    N=\left\lceil\log_2(2W_0 m)\right\rceil,\qquad
    M=\left\lceil\frac{2N}{\varepsilon}\right\rceil.
  \]
  Use the first $M$ primes as moduli.

  Compute a planar embedding of $G$, choose a root $s$, and let $L_i=\{v\in V(G):\dist_G(s,v)=i\}$ be the $i$-th BFS layer rooted at $s$.
  Every edge of $G$ has endpoints in the same layer or in two consecutive layers.
  For each $j\in\{0,\dots,k-1\}$, define
  \[
    B_j=\bigcup_{i\equiv j\pmod k}L_i,
    \qquad
    H_j=G-B_j.
  \]

  Let $C$ be a component of $H_j$, and let $i_C^{\min}$ and $i_C^{\max}$ be the smallest and largest layer indices met by $C$.
  No deleted layer has index between $i_C^{\min}$ and $i_C^{\max}$, so $i_C^{\max}-i_C^{\min}\leq k-2$.
  Every layer from $L_{i_C^{\min}}$ through $L_{i_C^{\max}}$ is contained in $H_j$, and any edge of $G[L_{i_C^{\min}}\cup\dots\cup L_{i_C^{\max}}]$ from $C$ to another vertex would also be an edge of $H_j$.
  Thus $C$ is a connected component of $G[L_{i_C^{\min}}\cup\dots\cup L_{i_C^{\max}}]$, and \cref{lem:planar-modular-strip} applies to $C$.
  For $v\in V(H_j)$, let $q_j(v)$ be the number of edges of $G$ joining $v$ to a vertex of $B_j$.

  For every $j\in\{0,\dots,k-1\}$, perform the preprocessing from \cref{lem:planar-modular-strip} for each component of $H_j$ that contains an edge.
  For every one of the first $M$ primes $p$, run the lemma's modular dynamic program on each such component using $q(v)=q_j(v)$.
  Combining the component solutions gives an edge-weighting $x_{j,p}:E(H_j)\rightarrow\{a,b\}$.
  For every vertex $v$ of $H_j$ that is not isolated, it also gives a residue $r_{j,p}(v)\in\mathcal R_p(q_j(v))$ and an integer $t_{j,p}(v)\in\{0,\dots,q_j(v)\}$ satisfying $r_{j,p}(v)\equiv q_j(v)a+t_{j,p}(v)(b-a)\pmod p$.
  If $v$ is isolated in $H_j$, set $t_{j,p}(v)=0$ and let $r_{j,p}(v)$ be the residue of $q_j(v)a$ modulo $p$.

  Extend $x_{j,p}$ to an edge-weighting $\widehat w_{j,p}$ of $G$.
  Keep the weights $x_{j,p}$ on $E(H_j)$.
  For each $v\in V(H_j)$, assign weight $b$ to exactly $t_{j,p}(v)$ of the edges from $v$ to $B_j$ and weight $a$ to the remaining such edges.
  These choices are independent for different vertices of $H_j$, because every edge between $H_j$ and $B_j$ has exactly one endpoint in $H_j$.
  Assign weight $a$ to every edge with both endpoints in $B_j$.
  By the definition of $r_{j,p}(v)$ and $t_{j,p}(v)$, every $v\in V(H_j)$ satisfies
  \[
    z_{\widehat w_{j,p}}(v)
    \equiv
    r_{j,p}(v)
    +\sum_{e\in E_{H_j}(v)}x_{j,p}(e)
    \pmod p.
  \]
  Therefore every edge counted by a modular dynamic program is proper under $\widehat w_{j,p}$, because equal integer labels would have equal residues modulo $p$.
  The algorithm constructs all $kM$ candidate weightings $\widehat w_{j,p}$ and returns one with the largest number of proper edges in $G$.

  We next prove that one candidate attains the claimed approximation ratio.
  Let $w^*$ be an optimal edge-weighting of $G$, and let $F^*$ be its set of proper edges.
  For $e=uv\in F^*$, define $h_e=|z_{w^*}(u)-z_{w^*}(v)|$.
  Since $e$ is proper and every edge weight has absolute value at most $W_0$,
  \[
    1\leq h_e
    \leq |z_{w^*}(u)|+|z_{w^*}(v)|
    \leq W_0\bigl(d_G(u)+d_G(v)\bigr)
    \leq 2W_0 m.
  \]

  Fix $e\in F^*$.
  Each endpoint of $e$ belongs to exactly one of $B_0,\dots,B_{k-1}$, so at most two values of $j$ delete an endpoint of $e$.
  These values account for at most $2M$ of the $kM$ pairs $(j,p)$.
  Moreover, $h_e$ has at most $N$ distinct prime divisors, since the product of any $N+1$ distinct prime divisors would divide $h_e$ and be at least $2^{N+1}>2W_0m$, contradicting $h_e\leq 2W_0m$.
  The primes dividing $h_e$ therefore account for at most another $kN$ pairs.
  Consequently, for at least $kM-2M-kN$ pairs $(j,p)$, both endpoints of $e$ lie in $H_j$ and $p$ does not divide $h_e$.

  Count the triples $(e,j,p)$ with $e\in F^*$ for which these two conditions hold.
  Summing the preceding bound over $e\in F^*$ and then averaging over the $kM$ pairs gives a pair $(j,p)$ for which at least
  \[
    \left(1-\frac{2}{k}-\frac{N}{M}\right)|F^*|
    \geq
    (1-\varepsilon)\OPT_{\mathrm{EW}}(G)
  \]
  edges of $F^*$ have both endpoints in $H_j$ and satisfy $p\nmid h_e$.
  The inequality uses $2/k\leq\varepsilon/2$ and $N/M\leq\varepsilon/2$.

  Fix such a pair $(j,p)$.
  On each component $C$ of $H_j$ that contains an edge, use the restriction of $w^*$ as the internal edge-weighting.
  For every vertex $v$ of $C$, let $t^*(v)$ be the number of edges from $v$ to $B_j$ that have weight $b$ under $w^*$, and let $r^*(v)$ be the residue in $\{0,\dots,p-1\}$ satisfying $r^*(v)\equiv q_j(v)a+t^*(v)(b-a)\pmod p$.
  Then $r^*(v)\in\mathcal R_p(q_j(v))$ and
  \[
    r^*(v)+\sum_{e\in E_{H_j}(v)}w^*(e)
    \equiv z_{w^*}(v)\pmod p.
  \]
  On each component, these choices form a feasible solution of its modular subproblem.
  If $e=uv\in F^*$ has both endpoints in $H_j$ and $p\nmid h_e$, then $z_{w^*}(u)\not\equiv z_{w^*}(v)\pmod p$, so the feasible solution for its component counts $e$.
  The exact dynamic program for each component counts at least as many edges as this feasible solution, and every edge counted by a dynamic program is proper after extension to $\widehat w_{j,p}$.
  Hence the weighting returned by the algorithm has at least $(1-\varepsilon)\OPT_{\mathrm{EW}}(G)$ proper edges.

  We finish with the running-time bound.
  For each $j$, there are at most $m$ components of $H_j$ that contain an edge.
  Since $G$ is connected and has at most $m+1$ vertices, all preprocessing steps from \cref{lem:planar-modular-strip}, over all $k$ choices of $j$, take $O(km^2)$ time.
  For a fixed choice of $j$ and $p$, the sum of $|E(C)|+1$ over the components $C$ of $H_j$ that contain an edge is at most $2m$.
  The lemma therefore bounds the total number of arithmetic operations for their dynamic programs by $O(p^{13k+12}m)$.
  The standard estimate that the $n$-th prime is $O(n\log(2n))$ implies that each of the first $M$ primes is $O(M\log(2M))$~\cite{rosser1962approximate}.
  The sieve of Eratosthenes generates these primes within the final running-time bound.
  Since there are $kM$ pairs $(j,p)$, the complete algorithm uses $O\bigl(kM\bigl(M\log(2M)\bigr)^{13k+12}m\bigr)$ arithmetic operations.
  Because $a$ and $b$ are fixed, $k=O(1/\varepsilon)$ and $M=O((1/\varepsilon)(1+\log m))$.
  Moreover, $\log(2M)=O(\log(2/\varepsilon)+\log(1+\log m))$.
  The preceding bound is therefore $2^{O((1/\varepsilon)\log(2/\varepsilon))}(1+\log m)^{O(1/\varepsilon)}m$.
  For all $m\geq 1$ and $\ell\geq 1$, the elementary bound $(1+\log m)^\ell\leq (2\ell)^\ell m$ holds.
  Consequently, the factor $(1+\log m)^{O(1/\varepsilon)}$ is at most $2^{O((1/\varepsilon)\log(2/\varepsilon))}m$, so the algorithm uses $2^{O((1/\varepsilon)\log(2/\varepsilon))}m^2$ arithmetic operations.
  This bound also covers the $O(km^2)$ preprocessing term.
  Since $a$ and $b$ are fixed and all integers used by the algorithm have polynomial bit length, this gives a running time of $2^{O((1/\varepsilon)\log(2/\varepsilon))}m^{O(1)}$ in the standard bit model.
\end{proof}

A fully polynomial-time approximation scheme (FPTAS) is a polynomial-time approximation scheme (PTAS) whose running time is polynomial in both the input size and $1/\varepsilon$.
Unless $\mathrm{P}=\mathrm{NP}$, no such scheme exists for \textsc{Max $\{a,b\}$-edge-weighting} on simple cubic planar graphs.
Since the objective counts proper edges, all feasible objective values, including $\OPT_{\mathrm{EW}}(G)$, are integers.
Run a hypothetical FPTAS on a graph with $m\geq 1$ edges using $\varepsilon=1/(m+1)$.
It would return a weighting with an integer number $A$ of proper edges satisfying $\OPT_{\mathrm{EW}}(G)-1<A\leq\OPT_{\mathrm{EW}}(G)$, and hence $A=\OPT_{\mathrm{EW}}(G)$.
Testing whether $A=m$ would then decide the NP-complete problem from \cref{thm:planar-ab-weighting} in polynomial time.

\section{APX-completeness for simple cubic graphs}\label{sec:apx-completeness}

We first record a simple approximation algorithm for loopless multigraphs, and then prove APX-completeness for simple cubic graphs.
In the approximation theorem, parallel edges are treated as distinct.

\begin{theorem}\label{thm:half-approx-ew}
  For every fixed pair of distinct integers $a$ and $b$, there is a polynomial-time $1/2$-approximation algorithm for \textnormal{\textsc{Max $\{a,b\}$-edge-weighting}} on loopless multigraphs.
\end{theorem}
\begin{proof}
  Fix distinct integers $a$ and $b$, and let $G$ be a loopless multigraph.
  Order its edges arbitrarily as $e_1,\dots,e_m$.
  For an edge $e=uv\in E(G)$, define $R(e)=E_G(u)\mathbin{\triangle}E_G(v)$, where $\triangle$ denotes symmetric difference.
  The edges joining $u$ and $v$, including all copies parallel to $e$, belong to both incident-edge sets and therefore do not belong to $R(e)$.
  Their weights contribute equally to $z_w(u)$ and $z_w(v)$, so the status of $e$ depends only on the weights of the edges in $R(e)$.
  In particular, if $R(e)=\emptyset$, then $e$ is improper under every weighting.
  Let $E'=\{e\in E(G):R(e)\neq\emptyset\}$.
  Since no edge outside $E'$ can be proper, $\OPT_{\mathrm{EW}}(G)\leq |E'|$.

  For each $e\in E'$, let $j(e)$ be the largest index $i$ such that $e_i\in R(e)$.
  We construct a weighting by processing $e_1,\dots,e_m$ in this order.
  Suppose weights have already been chosen on $e_1,\dots,e_{i-1}$, and let $A_i=\{e\in E':j(e)=i\}$.
  For each $c\in\{a,b\}$, let $N_i(c)$ be the number of edges in $A_i$ that would be proper if $w(e_i)=c$.
  We choose a value of $c$ that maximizes $N_i(c)$.

  This choice makes at least $|A_i|/2$ edges of $A_i$ proper.
  To see this, fix an edge $e=uv\in A_i$ and put $f=e_i$.
  Since $f\in R(e)$, it is incident with exactly one of $u$ and $v$.
  All other edges in $R(e)$ have already been assigned weights.
  Write
  \[
    X=\sum_{g\in E_G(u)\setminus E_G(v)}w(g),
    \qquad
    Y=\sum_{g\in E_G(v)\setminus E_G(u)}w(g).
  \]
  The edge $e$ is proper exactly when $X\neq Y$.
  Since $f$ occurs in exactly one of these two sums, at most one of the choices $w(f)=a$ and $w(f)=b$ makes $X=Y$.
  Thus $e$ is proper for at least one choice, and hence $N_i(a)+N_i(b)\geq |A_i|$.
  Later choices cannot change the status of an edge $e\in A_i$, because all edges in $R(e)$ have index at most $i$.

  The sets $A_1,\dots,A_m$ partition $E'$.
  Therefore the final weighting has at least
  \[
    \sum_{i=1}^m\frac{|A_i|}{2}
    =
    \frac{|E'|}{2}
    \geq
    \frac{\OPT_{\mathrm{EW}}(G)}{2}
  \]
  proper edges.
  The sets $R(e)$, the indices $j(e)$, and the values $N_i(a)$ and $N_i(b)$ can all be computed directly in polynomial time.
\end{proof}

\begin{theorem}\label{thm:apx-complete-cubic}
  For every fixed pair of distinct integers $a$ and $b$, the problem \textnormal{\textsc{Max $\{a,b\}$-edge-weighting}} is APX-complete even when restricted to simple cubic graphs.
\end{theorem}
\begin{proof}
  Fix distinct integers $a$ and $b$.
  Membership in APX follows from \cref{thm:half-approx-ew}.

  By \cref{prop:regular-weight-normalization}, it remains to prove APX-hardness for \textsc{Max $\{0,1\}$-edge-weighting} on simple cubic graphs.

  We reduce from \textsc{Max Cut} restricted to simple cubic graphs, which is APX-hard~\cite{alimonti1997hardness}.
  Let $G=(V,E)$ be a simple cubic graph, and write $\OPT_{\mathrm{CUT}}(G)$ for the maximum cut value of $G$.
  For $S\subseteq V$, let $\delta_G(S)$ denote the set of edges of $G$ with exactly one endpoint in $S$.
  We construct a cubic graph $G'$ by replacing every vertex and edge of $G$ with constant-size gadgets.
  From any cut of $G$, we construct a $\{0,1\}$-edge-weighting of $G'$.
  Conversely, from any weighting of $G'$, we construct a cut of $G$ and associate every original edge that does not cross this cut with a distinct improper edge of $G'$.

  We reuse the rhombus gadget and the $2$- and $3$-port synchronizers from \cref{sec:abWeightingNPC}, shown in \cref{fig:rhombus-gadget,fig:cubic-port-synchronizers}.
  By \cref{claim:port-synchronizer}, if all port edges of either synchronizer are assigned a common weight $s\in\{0,1\}$, then the remaining edges can be weighted so that every edge lying in a rhombus gadget is proper and every junction vertex has label $3s$.
  Conversely, if the weights of the port edges are not all equal, then some edge lying in a rhombus gadget is improper.

  We now construct $G'$ from $G$.
  For every vertex $v\in V$, insert a disjoint copy $A_v$ of a $3$-port synchronizer, called the \emph{vertex gadget} for $v$.
  Associate the three port edges of $A_v$ with the three edges of $G$ incident with $v$.

  For every original edge $e=uv\in E$, let $p_{u,e}$ and $p_{v,e}$ be the junction vertices of $A_u$ and $A_v$, respectively, whose port edges correspond to $e$.
  Add two new vertices $x_e,y_e$.
  Use the two corresponding port edges of $A_u$ and $A_v$, identifying their degree-one endpoints with $x_e$ and $y_e$, as the edges $p_{u,e}x_e$ and $y_ep_{v,e}$.
  Add the edge $x_ey_e$ between the two new vertices.
  We call $x_ey_e$ the \emph{middle edge} corresponding to $e$.
  To make $x_e$ and $y_e$ have degree $3$, add a disjoint copy $B_e$ of a $2$-port synchronizer.
  Let $q_{x,e}$ and $q_{y,e}$ be the two junction vertices of $B_e$.
  Use the two port edges of $B_e$, identifying their degree-one endpoints with $x_e$ and $y_e$, as the edges $x_eq_{x,e}$ and $y_eq_{y,e}$.
  After these identifications, both $x_e$ and $y_e$ have degree $3$.
  By \cref{claim:port-synchronizer}, if $w(x_eq_{x,e})\neq w(y_eq_{y,e})$, then some edge in a rhombus gadget of $B_e$ is improper.
  If these two edges have the same weight, the remaining edges of $B_e$ can be weighted so that every edge in its rhombus gadgets is proper.
  See \cref{fig:cubic-edge-replacement}.

  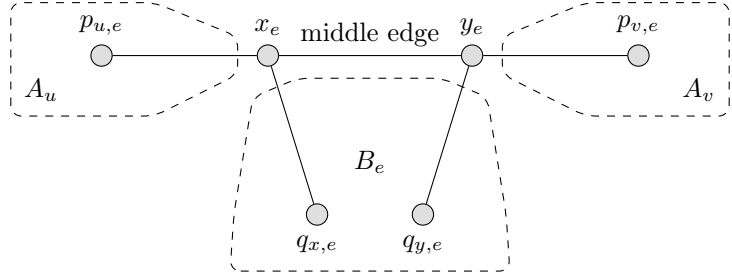
\begin{figure}[H]
    \centering
    \begin{tikzpicture}[scale=1]
      \tikzset{gadgetbox/.style={draw,dashed,rounded corners,line width=.4pt}}
      \node[figure vertex,label=above:$p_{u,e}$] (pu) at (-3.55,1.15) {};
      \node[figure vertex,label=above:$p_{v,e}$] (pv) at (3.55,1.15) {};
      \node[figure vertex,label=above:$x_e$] (x) at (-1.35,1.15) {};
      \node[figure vertex,label=above:$y_e$] (y) at (1.35,1.15) {};
      \node[figure vertex,label=below:$q_{x,e}$] (qx) at (-.7,-0.9518) {};
      \node[figure vertex,label=below:$q_{y,e}$] (qy) at (.7,-0.9518) {};
      \draw[figure edge] (pu)--(x);
      \draw[figure edge] (x)--(y) node[midway,above] {middle edge};
      \draw[figure edge] (y)--(pv);
      \draw[figure edge] (x)--(qx);
      \draw[figure edge] (y)--(qy);
      \draw[gadgetbox] (-4.75,.35) -- (-4.75,1.85) -- (-2.85,1.85) -- (-1.75,1.4) -- (-1.75,.9) -- (-2.85,.35) -- cycle;
      \draw[gadgetbox] (4.75,.35) -- (4.75,1.85) -- (2.85,1.85) -- (1.75,1.4) -- (1.75,.9) -- (2.85,.35) -- cycle;
      \draw[gadgetbox] (-1.85,-1.7) -- (1.85,-1.7) -- (1.8,-.95) -- (1.55,.6) -- (.95,.85) -- (-.95,.85) -- (-1.55,.6) -- (-1.8,-.95) -- cycle;
      \node at (-4.35,.7) {$A_u$};
      \node at (4.35,.7) {$A_v$};
      \node at (0,-.25) {$B_e$};
    \end{tikzpicture}
    \caption{The part of $G'$ replacing an edge $e=uv$ of $G$.
    The dashed regions indicate the vertex gadgets $A_u,A_v$ and the $2$-port synchronizer $B_e$; the displayed edges incident with $p_{u,e}$, $p_{v,e}$, $q_{x,e}$, and $q_{y,e}$ are their port edges.}
    \label{fig:cubic-edge-replacement}
  \end{figure}

  This completes the construction of $G'$.
  The graph $G'$ is simple and cubic.
  It is simple because the gadget copies are otherwise vertex-disjoint and the three neighbors of each vertex $x_e$ or $y_e$ are distinct.
  Every internal vertex of a rhombus gadget has degree $3$, every junction vertex has two terminal edges of rhombus gadgets and one port edge, and each of $x_e$ and $y_e$ is incident with one edge to a vertex gadget, the middle edge $x_ey_e$, and one port edge of $B_e$.

  Let $K=\bigl|E(G')\setminus\{x_ey_e:e\in E\}\bigr|$ be the number of edges of $G'$ other than the middle edges.
  We now prove the value identity
  \[
    \OPT_{\mathrm{EW}}(G')=K+\OPT_{\mathrm{CUT}}(G).
  \]
  For the lower bound, we construct a $\{0,1\}$-edge-weighting of $G'$ from any cut of $G$.
  For the reverse inequality, we construct a cut from any weighting and associate each original edge that does not cross this cut with a distinct improper edge of $G'$.

  We first show that every cut $S\subseteq V$ gives a weighting $w_S$ of $G'$ with exactly $K+|\delta_G(S)|$ proper edges.
  Define $s_v=0$ if $v\in S$, and define $s_v=1$ otherwise.
  In every vertex gadget $A_v$, assign weight $s_v$ to all three port edges and apply \cref{claim:port-synchronizer} to weight the remaining edges.
  For every $e\in E$, assign weight $0$ to both port edges of $B_e$ and apply \cref{claim:port-synchronizer} to weight its remaining edges.
  Finally, set every middle edge $x_ey_e$ to have weight $1$.

  For an edge $e=uv\in E$, the labels under $w_S$ of the two new vertices are $z_{w_S}(x_e)=s_u+1$ and $z_{w_S}(y_e)=s_v+1$.
  The junction vertex of $A_u$ adjacent to $x_e$ has label $3s_u$, and the junction vertex $q_{x,e}$ of $B_e$ has label $0$.
  Therefore the two edges incident with $x_e$ other than $x_ey_e$ are proper.
  The same argument applies at $y_e$.
  All edges in the rhombus gadgets of the vertex gadgets and the synchronizers $B_e$ are proper by construction.
  Hence the only edge whose properness depends on $S$ is the middle edge $x_ey_e$.
  It is proper exactly when $s_u+1\neq s_v+1$, which is equivalent to $uv\in\delta_G(S)$.
  Therefore $w_S$ has exactly $K+|\delta_G(S)|$ proper edges.
  Taking $S$ to be a maximum cut gives $\OPT_{\mathrm{EW}}(G')\geq K+\OPT_{\mathrm{CUT}}(G)$.

  We now prove the reverse inequality.
  Starting from an arbitrary weighting of $G'$, we recover a cut of $G$ from the majority values of the vertex gadgets.
  Let $w$ be an arbitrary $\{0,1\}$-edge-weighting of $G'$.
  For every vertex $v\in V$, consider the weights of the three port edges of the vertex gadget $A_v$.
  Let $s_v$ be their majority value, which is unique because there are three binary weights, and define $S_w=\{v\in V:s_v=0\}$.
  We prove that the number of proper edges under $w$ is at most $K+|\delta_G(S_w)|$.
  Equivalently, it is enough to prove that $w$ has at least $|E|-|\delta_G(S_w)|$ improper edges.
  Since $S_w=\{v\in V:s_v=0\}$, an original edge $e=uv$ crosses $S_w$ exactly when $s_u\neq s_v$.
  Thus $|E|-|\delta_G(S_w)|$ is the number of original edges $e=uv$ that do not cross $S_w$, namely those with $s_u=s_v$.
  We therefore construct an injection from these original edges to improper edges of $G'$.
  For each such edge, either its middle edge is improper or, if the middle edge is proper, one of the adjacent synchronizers contains an improper edge.

  Fix an edge $e=uv\in E$ with $s_u=s_v$.
  The labels of $x_e$ and $y_e$ are
  \[
    z_w(x_e)=w(p_{u,e}x_e)+w(x_eq_{x,e})+w(x_ey_e),\qquad
    z_w(y_e)=w(y_ep_{v,e})+w(y_eq_{y,e})+w(x_ey_e).
  \]
  If $x_ey_e$ is improper under $w$, associate $e$ with the improper edge $x_ey_e$.

  Assume now that $x_ey_e$ is proper.
  Since the weight of $x_ey_e$ appears in both labels, we have
  \[
    w(p_{u,e}x_e)+w(x_eq_{x,e})
    \neq
    w(y_ep_{v,e})+w(y_eq_{y,e}).
  \]
  Since $s_u=s_v$, this inequality is impossible if the three equalities
  \[
    w(p_{u,e}x_e)=s_u,
    \qquad
    w(y_ep_{v,e})=s_v,
    \qquad
    w(x_eq_{x,e})=w(y_eq_{y,e})
  \]
  all hold.
  Hence at least one of
  \[
    w(p_{u,e}x_e)\neq s_u,
    \qquad
    w(y_ep_{v,e})\neq s_v,
    \qquad
    w(x_eq_{x,e})\neq w(y_eq_{y,e})
  \]
  holds.
  If $w(p_{u,e}x_e)\neq s_u$, then the weights of the port edges of $A_u$ are not all equal.
  By \cref{claim:port-synchronizer}, $A_u$ contains an improper edge in one of its rhombus gadgets; associate $e$ with one such edge.
  Otherwise, if $w(y_ep_{v,e})\neq s_v$, then the weights of the port edges of $A_v$ are not all equal, so \cref{claim:port-synchronizer} gives an improper edge in one of its rhombus gadgets; associate $e$ with one such edge.
  In the remaining case, $w(x_eq_{x,e})\neq w(y_eq_{y,e})$, so \cref{claim:port-synchronizer} gives an improper edge in a rhombus gadget of $B_e$; associate $e$ with one such edge.

  These associations can be chosen injectively.
  A middle edge $x_ey_e$ can be associated only with its own original edge $e$, and an edge in a rhombus gadget of $B_e$ can be used only for $e$.
  The only possible conflict is therefore inside a vertex gadget.
  For a fixed vertex gadget $A_v$, an original edge can be associated with an improper edge in $A_v$ only when its corresponding port is the unique one whose weight differs from the majority value $s_v$.
  Of the three binary weights on the port edges, at most one differs from the value that occurs at least twice.
  Thus at most one original edge is associated with any fixed vertex gadget, and the choices inside vertex gadgets can be made without conflict.
  Since the middle edges, the edges in the rhombus gadgets of the synchronizers $B_e$, and the edges in the rhombus gadgets of the vertex gadgets form pairwise disjoint sets, distinct edges $e=uv$ with $s_u=s_v$ are associated with distinct improper edges of $G'$.

  It follows that $w$ has at least $|E|-|\delta_G(S_w)|$ improper edges.
  Consequently, the number of proper edges under $w$ is at most
  \[
    |E(G')|-\bigl(|E|-|\delta_G(S_w)|\bigr)
    =K+|\delta_G(S_w)|.
  \]
  Apply this bound to an optimal weighting $w$ of $G'$.
  Then
  \[
    \OPT_{\mathrm{EW}}(G')\leq K+|\delta_G(S_w)|\leq K+\OPT_{\mathrm{CUT}}(G).
  \]
  Together with the lower bound proved above, this gives
  \[
    \OPT_{\mathrm{EW}}(G')=K+\OPT_{\mathrm{CUT}}(G).
  \]

  Both the construction of $G'$ and the recovery of $S_w$ can be carried out in polynomial time.
  To complete the $L$-reduction, it remains to verify its two inequalities with constants $\alpha=65$ and $\beta=1$.
  For the first inequality, we bound $\OPT_{\mathrm{EW}}(G')$ in terms of $\OPT_{\mathrm{CUT}}(G)$.
  The vertex gadgets contribute $3\cdot 7|V|=21|V|$ edges from rhombus gadgets.
  For every original edge $e\in E$, the construction contributes the two edges $p_{u,e}x_e$ and $y_ep_{v,e}$, the middle edge $x_ey_e$, the fourteen edges from rhombus gadgets in $B_e$, and the two edges $x_eq_{x,e}$ and $y_eq_{y,e}$, for a total of $19$ edges.
  Since $G$ is cubic, $3|V|=2|E|$, and hence $|E(G')|=21|V|+19|E|=33|E|$.
  There is one middle edge for each edge of $G$, so $K=32|E|$.
  Every graph has a cut containing at least half of its edges, and therefore $|E|\leq 2\OPT_{\mathrm{CUT}}(G)$.
  The value identity now gives
  \[
    \OPT_{\mathrm{EW}}(G')
    =32|E|+\OPT_{\mathrm{CUT}}(G)
    \leq 65\OPT_{\mathrm{CUT}}(G).
  \]

  For the second $L$-reduction inequality, let $w$ be any $\{0,1\}$-edge-weighting of $G'$, let $c_{\mathrm{EW}}(w)$ be its number of proper edges, and let $S_w$ be the cut obtained from the majority values of the vertex gadgets.
  The injection argument gives $c_{\mathrm{EW}}(w)\leq K+|\delta_G(S_w)|$.
  Together with the value identity, this implies
  \[
    \OPT_{\mathrm{CUT}}(G)-|\delta_G(S_w)|
    \leq
    \OPT_{\mathrm{EW}}(G')-c_{\mathrm{EW}}(w).
  \]
  Suppose that $w$ has at least $(1-\varepsilon')\OPT_{\mathrm{EW}}(G')$ proper edges for some $\varepsilon'>0$.
  Combining the last two inequalities yields
  \[
    |\delta_G(S_w)|
    \geq
    \OPT_{\mathrm{CUT}}(G)-\varepsilon'\OPT_{\mathrm{EW}}(G')
    \geq
    (1-65\varepsilon')\OPT_{\mathrm{CUT}}(G).
  \]
  Thus, for any $0<\varepsilon<1$, choosing $\varepsilon'=\varepsilon/65$ transforms a $(1-\varepsilon/65)$-approximate weighting of $G'$ into a $(1-\varepsilon)$-approximate cut of $G$.
  Since \textsc{Max Cut} on simple cubic graphs is APX-hard, this proves APX-hardness for \textsc{Max $\{0,1\}$-edge-weighting} on simple cubic graphs.
  The bijection in \cref{prop:regular-weight-normalization} preserves the number of proper edges, so APX-hardness also holds for \textsc{Max $\{a,b\}$-edge-weighting} for every fixed pair of distinct integers $a$ and $b$.
  Together with membership in APX, this completes the proof.
\end{proof}

\section{Extending partial edge-weightings}\label{sec:partialWeightings}
Thomassen, Wu, and Zhang proved in 2016 that the $\{1,2\}$-property can be decided in polynomial time for bipartite graphs~\cite{thomassen20163}, whereas the problem is NP-complete for arbitrary graphs~\cite{dudek2011complexity}.
Motivated by this result, we investigate whether a \emph{partial} $\{a,b\}$-edge-weighting, which prescribes the weights on a subset of the edges, can be extended to a proper $\{a,b\}$-edge-weighting in bipartite graphs.
We first show that this problem is NP-complete for every fixed pair of distinct integers $a$ and $b$, even for simple cubic planar bipartite graphs under the additional restriction that each component of the subgraph formed by the prescribed edges is a path of length $6$.
We then give a polynomial-time algorithm when the graph is a tree.
When $a=0$ and $b=1$, this gives an alternative to the polynomial-time decision algorithm obtained from Lyngsie's characterization of all trees without a proper $\{0,1\}$-edge-weighting~\cite{lyngsie2018neighbour}.
The tree algorithm also gives an alternative polynomial-time algorithm for the antifactor problem on trees, introduced by Lov\'{a}sz~\cite{lovasz1973antifactors}.

\begin{theorem}\label{thm:extProbHard}
  For every fixed pair of distinct integers $a$ and $b$, deciding whether a partial $\{a,b\}$-edge-weighting can be extended to a proper $\{a,b\}$-edge-weighting is NP-complete, even for simple cubic planar bipartite graphs.
  Moreover, NP-hardness holds when each component of the subgraph formed by the prescribed edges is a path of length $6$, all of whose edges have the same prescribed weight.
\end{theorem}
\begin{proof}
  The problem is in NP because a proposed extension can be checked in polynomial time.
  For NP-hardness, we reduce from \textsc{Cubic Planar Monotone 1-in-3-SAT}.
  An instance is a CNF formula $\Phi$ in which every clause contains exactly three distinct unnegated variables and every variable occurs exactly three times.
  The \emph{formula graph} of $\Phi$ is the bipartite graph with one vertex for each variable and each clause and one edge for each variable occurrence, joining the corresponding variable and clause vertices.
  Only instances with a planar formula graph are allowed.
  The question is whether there is a truth assignment under which exactly one variable in every clause is true.
  We call such an assignment a $1$-in-$3$ assignment.
  This restricted problem is NP-complete~\cite{moore2001hard}.
  Let $G$ be the formula graph, and fix a planar embedding of $G$.
  The graph $G$ is simple, cubic, planar, and bipartite.

  By \cref{prop:regular-weight-normalization}, it is enough to prove NP-hardness for $\{0,1\}$-edge-weightings.
  For any cubic partial $\{0,1\}$-edge-weighting, the bijection in that proposition maps proper extensions exactly to proper extensions after every prescribed $0$ is replaced by $a$ and every prescribed $1$ by $b$.
  It maps a prescribed path whose edges all have weight $0$ or all have weight $1$ to one whose edges all have weight $a$ or all have weight $b$, respectively.
  All edge weights in the remainder of the proof belong to $\{0,1\}$.

  The reduction replaces every variable vertex and every clause vertex of $G$ by a copy of the same graph $J$ and joins the copies by unprescribed occurrence edges corresponding to the edges of $G$.
  In a variable gadget, the prescribed weights force the three occurrence edges to have the same weight and make the label at each terminal equal to that weight.
  In a clause gadget, the prescribed weights permit a proper weighting if and only if exactly one of the three occurrence edges has weight $1$, provided that the label at the other endpoint of each occurrence edge equals its weight.
  We define the two gadgets and establish these statements before assembling the graph.

  Let $J$ be the bipartite graph with bipartition $U_J=\{u_0,u_1,u_2,u_3,u_4,u_5\}$ and $V_J=\{v_0,v_1,v_2,v_3,v_4\}$.
  Its edges consist of the cycle $u_0v_3u_4v_0u_3v_4u_0$, the path $u_2v_1u_5v_2u_1$, and the five additional edges $v_3u_2$, $u_4v_1$, $v_0u_5$, $u_3v_2$, and $v_4u_1$.
  The vertices $u_0,u_1,u_2$ are the terminals of $J$.
  They have degree two, while every other vertex has degree three.
  Adding a new vertex $v_5$ adjacent to $u_0,u_1,u_2$ produces the hexagonal prism.
  Deleting $v_5$ from a planar embedding of the prism merges the faces incident with $v_5$ into one face containing $u_0,u_1,u_2$.
  Hence $J$ has a planar embedding in which its three terminals lie on one face.

  For the local analysis, attach one unprescribed occurrence edge to each terminal of $J$.
  The variable gadget $\mathcal V$ is obtained by prescribing weight $0$ on $u_0v_3$, $u_0v_4$, $u_1v_2$, $u_1v_4$, $u_2v_1$, and $u_2v_3$.
  The clause gadget $\mathcal C$ is obtained by prescribing weight $1$ on $u_0v_3$, $u_0v_4$, $u_3v_2$, $u_3v_4$, $u_4v_0$, and $u_4v_3$.
  All other edges of $J$ remain unprescribed in both gadgets.

  \begin{claim}\label{claim:extension-variable-gadget}
    Let $x_i\in\{0,1\}$ be the weight of the occurrence edge incident with $u_i$ in $\mathcal V$.
    The unprescribed edges of $J$ can be weighted so that every edge of $J$ is proper if and only if $x_0=x_1=x_2$.
    In every such weighting, $z(u_i)=x_i$ for $i\in\{0,1,2\}$.
  \end{claim}
  \begin{proof}
    Both edges of $J$ incident with each terminal have prescribed weight $0$, so $z(u_i)=x_i$ for $i\in\{0,1,2\}$.
    Put $y_{43}=w(u_4v_3)$ and $y_{34}=w(u_3v_4)$.
    The labels of $v_3$ and $v_4$ are $y_{43}$ and $y_{34}$, respectively.
    Properness of $u_0v_3$ and $u_2v_3$ gives $x_0\neq y_{43}$ and $x_2\neq y_{43}$.
    Since these three values belong to $\{0,1\}$, we obtain $x_0=x_2$.
    The same argument applied to $u_0v_4$ and $u_1v_4$ gives $x_0=x_1$.

    It remains to weight the unprescribed edges when the three occurrence edges have a common weight.
    If $x_0=x_1=x_2=0$, assign weight $1$ to $u_3v_0$, $u_3v_2$, $u_3v_4$, $u_4v_0$, $u_4v_1$, and $u_4v_3$, and weight $0$ to every other unprescribed edge of $J$.
    The labels of $u_0,\dots,u_5$ and $v_0,\dots,v_4$ are then $(0,0,0,3,3,0)$ and $(2,1,1,1,1)$, respectively.
    If $x_0=x_1=x_2=1$, assign weight $1$ to $u_3v_0$, $u_4v_0$, and $u_5v_0$, and weight $0$ to every other unprescribed edge of $J$.
    The labels of $u_0,\dots,u_5$ and $v_0,\dots,v_4$ are then $(1,1,1,1,1,1)$ and $(3,0,0,0,0)$, respectively.
    In each case, no label on $U_J$ appears on $V_J$, so every edge of $J$ is proper.
  \end{proof}

  \begin{claim}\label{claim:extension-clause-gadget}
    Let $x_i\in\{0,1\}$ be the weight of the occurrence edge incident with $u_i$ in $\mathcal C$.
    Suppose that the other endpoint of the occurrence edge at $u_i$ has label $x_i$ for each $i\in\{0,1,2\}$.
    The unprescribed edges of $J$ can be weighted so that every edge of $J$ and all three occurrence edges are proper if and only if $x_0+x_1+x_2=1$.
  \end{claim}
  \begin{proof}
    For every unprescribed edge $u_iv_j$ of $J$, write $y_{ij}=w(u_iv_j)$.
    Since $u_3v_2$, $u_3v_4$, and $u_4v_0$ have prescribed weight $1$, the endpoints of $u_3v_0$ have labels $y_{30}+2$ and $y_{30}+1+y_{50}$.
    Properness of $u_3v_0$ therefore gives $y_{50}=0$.

    The vertex labels are
    \[
    \begin{array}{llll}
      z(u_0)=x_0+2,
      & z(u_1)=x_1+y_{12}+y_{14},
      & z(u_2)=x_2+y_{21}+y_{23},
      & z(v_3)=y_{23}+2,\\
      z(u_3)=y_{30}+2,
      & z(u_4)=y_{41}+2,
      & z(u_5)=y_{51}+y_{52},
      & z(v_4)=y_{14}+2,\\
      z(v_0)=y_{30}+1,
      & z(v_1)=y_{21}+y_{41}+y_{51},
      & z(v_2)=y_{12}+1+y_{52}.
    \end{array}
    \]
    Properness of $u_0v_3$, $u_0v_4$, $u_3v_4$, and $u_4v_3$ gives $y_{23}=y_{14}=1-x_0$ and $y_{30}=y_{41}=x_0$.

    Suppose first that $x_0=1$.
    Then $y_{14}=y_{23}=0$.
    The occurrence edge at $u_1$ has endpoint labels $x_1$ and $x_1+y_{12}$, so its properness forces $y_{12}=1$.
    The edge $u_1v_4$ then has endpoint labels $x_1+1$ and $2$, and hence $x_1=0$.
    The same argument at $u_2$ gives $y_{21}=1$ and $x_2=0$.

    Suppose now that $x_0=0$.
    We have $y_{14}=y_{23}=1$ and $y_{30}=y_{41}=0$.
    The endpoints of $u_1v_2$ have labels $x_1+y_{12}+1$ and $y_{12}+1+y_{52}$, so properness gives $y_{52}=1-x_1$.
    The endpoints of $u_3v_2$ have labels $2$ and $y_{12}+1+y_{52}$.
    Since $y_{12},y_{52}\in\{0,1\}$, properness of this edge gives $y_{12}=y_{52}$.
    If $y_{12}=0$, properness of $u_5v_2$ forces $y_{51}=0$, while if $y_{12}=1$, properness of $u_5v_0$ forces $y_{51}=1$.
    Hence $y_{51}=y_{12}$.
    The edge $u_5v_1$ now has endpoint labels $2y_{12}$ and $y_{21}+y_{12}$, so properness gives $y_{21}=1-y_{12}=x_1$.
    Consequently, $z(u_2)=x_1+x_2+1$, $z(v_3)=3$, and $z(v_1)=1$.
    The edge $u_2v_3$ excludes $x_1=x_2=1$, while $u_2v_1$ excludes $x_1=x_2=0$.
    Thus $x_1+x_2=1$.

    Conversely, suppose that $x_0+x_1+x_2=1$.
    Set $y_{50}=0$ and use the following values on the other unprescribed edges.
    \[
    \begin{array}{c|c|c|c}
      (x_0,x_1,x_2)&(y_{12},y_{14},y_{21},y_{23},y_{30},y_{41},y_{51},y_{52})&(z(u_0),\dots,z(u_5))&(z(v_0),\dots,z(v_4))\\ \hline
      (1,0,0)&(1,0,1,0,1,1,0,0)&(3,1,1,3,3,0)&(2,2,2,2,2)\\
      (0,1,0)&(0,1,1,1,0,0,0,0)&(2,2,2,2,2,0)&(1,1,1,3,3)\\
      (0,0,1)&(1,1,0,1,0,0,1,1)&(2,2,2,2,2,2)&(1,1,3,3,3)
    \end{array}
    \]
    In each row, no label on $U_J$ appears on $V_J$, so every edge of $J$ is proper.
    The first three entries of the $U_J$-label vector differ coordinatewise from $(x_0,x_1,x_2)$, so the three occurrence edges are proper as well.
  \end{proof}

  We now assemble the graph produced by the reduction.
  Choose pairwise disjoint disks around the vertices of the fixed planar embedding of $G$.
  Place a copy of $\mathcal V$ in each variable disk and a copy of $\mathcal C$ in each clause disk.
  For each copy of $J$, choose an embedding whose cyclic order of the three terminals on the boundary of the disk agrees with the cyclic order of the three incidences at the corresponding vertex of $G$.
  For every edge of $G$, join the corresponding variable and clause terminals by one unprescribed occurrence edge drawn along that edge.
  Since the terminals of $J$ lie on one face, the resulting graph is planar.
  Use the bipartition $(U_J,V_J)$ in the variable gadgets and its reverse in the clause gadgets.
  Every occurrence edge joins terminals $u_i$ in a variable copy and a clause copy, which belong to opposite parts under these choices, so the resulting graph is bipartite.
  Every vertex of $J$ other than the terminals has degree three, and every terminal has two incident edges in $J$ and one occurrence edge.
  Hence the resulting graph is cubic.
  The graph $J$ is simple, and distinct incidences use distinct terminals at both ends, so the resulting graph is simple.

  In every variable gadget, the prescribed edges form the path $v_1u_2v_3u_0v_4u_1v_2$, and all six edges have prescribed weight $0$.
  In every clause gadget, the prescribed edges form the path $v_2u_3v_4u_0v_3u_4v_0$, and all six edges have prescribed weight $1$.
  The gadget copies are vertex-disjoint and the occurrence edges are unprescribed.
  Consequently, each component of the subgraph formed by the prescribed edges is a path of length $6$, all of whose edges have the same prescribed weight.
  The construction has linear size and can be carried out in polynomial time.

  Suppose first that $\Phi$ has a $1$-in-$3$ assignment.
  Give an occurrence edge weight $1$ when its variable is true and weight $0$ otherwise.
  At every variable gadget, the three occurrence edges have the same weight, so \cref{claim:extension-variable-gadget} gives a proper weighting of all edges of $J$ and makes the label at each variable terminal equal to the weight of its occurrence edge.
  Every clause is incident with exactly one occurrence edge of weight $1$.
  \Cref{claim:extension-clause-gadget} therefore supplies weights for the remaining edges of every clause gadget and makes its three occurrence edges proper.
  This gives a proper extension of the partial weighting.

  Conversely, suppose that the constructed partial weighting has a proper extension.
  \Cref{claim:extension-variable-gadget} shows that the three occurrence edges at every variable gadget have the same weight.
  Assign the variable the value true when this common weight is $1$ and false when it is $0$.
  The label at the variable endpoint of each occurrence edge equals its weight.
  \Cref{claim:extension-clause-gadget} then shows that exactly one occurrence edge incident with every clause has weight $1$.
  Hence exactly one variable in every clause is true under the resulting assignment.

  This polynomial-time construction produces a partial $\{0,1\}$-edge-weighting of a simple cubic planar bipartite graph that has a proper extension if and only if $\Phi$ has a $1$-in-$3$ assignment.
  Hence the problem is NP-hard for $\{0,1\}$.
  By \cref{prop:regular-weight-normalization}, the same NP-hardness result holds for every fixed pair of distinct integers $a$ and $b$, and all edges of each prescribed path still have the same weight.
  Membership in NP completes the proof.
\end{proof}

\begin{corollary}\label{cor:liec-extension-hardness}
  It is NP-complete to decide whether a partial $2$-edge-coloring of a simple cubic planar bipartite graph can be extended to a locally irregular $2$-edge-coloring.
  Moreover, NP-hardness holds when each component of the subgraph formed by the prescribed edges is a path of length $6$, all of whose edges have the same prescribed color.
\end{corollary}
\begin{proof}
  Identify the two colors with $0$ and $1$.
  By \cref{prop:regular-weight-normalization}, a full $2$-edge-coloring of a cubic graph is locally irregular if and only if its $\{0,1\}$-edge-weighting is proper.
  Thus a partial $2$-edge-coloring extends to a locally irregular one if and only if the same edge assignment, viewed as a partial $\{0,1\}$-edge-weighting, extends to a proper weighting.
  The result follows from \cref{thm:extProbHard} with $a=0$ and $b=1$.
\end{proof}

Every cubic bipartite graph admits a locally irregular $2$-edge-coloring~\cite{baudon2015decomposing}.
In contrast, \cref{cor:liec-extension-hardness} shows that prescribing colors on disjoint monochromatic paths of length $6$ already makes the extension problem NP-complete on simple cubic planar bipartite graphs.

\subsection{Extending on trees}\label{sec:extTree}

\Cref{thm:extProbHard} shows that extending partial edge-weightings is NP-complete even for simple cubic planar bipartite graphs.
In this subsection, we give a polynomial-time algorithm that, for a tree and distinct integers $a$ and $b$, either extends a partial $\{a,b\}$-edge-weighting to a proper one or concludes that no extension exists.
For the unprescribed existence problem on trees, Madarasi recently gave a structural characterization for every pair of distinct real weights and a linear-time algorithm for every fixed pair~\cite{madarasi2026trees}.
The result below instead treats the extension problem, in which a partial $\{a,b\}$-edge-weighting may prescribe the weights of an arbitrary subset of edges.
For $a=0$ and $b=1$, the algorithm gives an alternative to the polynomial-time decision algorithm obtained from Lyngsie's characterization of all trees without a proper $\{0,1\}$-edge-weighting~\cite{lyngsie2018neighbour}.

\begin{theorem}\label{thm:treet}
  For every fixed pair of distinct integers $a$ and $b$, extending a partial $\{a,b\}$-edge-weighting to a proper $\{a,b\}$-edge-weighting is solvable in polynomial time for trees.
\end{theorem}
\begin{proof}
  Let $T=(V,E)$ be a tree with a partial $\{a,b\}$-edge-weighting.
  If $E=\emptyset$, then the empty weighting is a proper extension, so assume that $E\neq\emptyset$.
  For every edge $e$, let $W_e\subseteq\{a,b\}$ be its set of allowed weights, so $W_e$ is a singleton when the partial weighting prescribes the weight of $e$ and $W_e=\{a,b\}$ otherwise.
  Choose a leaf $v_0$ as the root and orient every edge away from $v_0$.
  For a vertex $v$, let $T_v$ be the subtree rooted at $v$.

  For each oriented edge $e=uv$, where $u$ is the parent of $v$, define a table $\mathcal T_e$.
  A pair $(z,x)\in\Z\times\{a,b\}$ belongs to $\mathcal T_e$ if $x\in W_e$ and all edges with both endpoints in $T_v$ can be weighted consistently with the partial weighting so that every such edge is proper and, after assigning weight $x$ to the parent edge $e=uv$, the full label of $v$ is $z$.
  The edge $e$ itself is not checked in this table because the label of its parent is determined higher in the rooted tree.

  Let $e_1,\dots,e_{d_T(v)-1}$ be the edges from $v$ to its children.
  A pair $(z,x)$ belongs to $\mathcal T_e$ if and only if one can choose $(z_i,x_i)\in\mathcal T_{e_i}$ for every child edge so that $z_i\neq z$ for all $i$ and $\sum_{i=1}^{d_T(v)-1}x_i=z-x$.
  For a fixed value of $z$, let $W_{e_i}(z)\subseteq W_{e_i}$ be the set of weights $x_i$ for which the child table $\mathcal T_{e_i}$ contains an entry $(z_i,x_i)$ with $z_i\neq z$.
  If some $W_{e_i}(z)$ is empty, then $(z,x)\notin\mathcal T_e$.
  Otherwise, solve
  \begin{align*}
    n_a+n_b&=d_T(v)-1,\\
    a n_a+b n_b&=z-x.
  \end{align*}
  Because $a\neq b$, this system has at most one integer solution $(n_a,n_b)$.
  The pair $(z,x)$ belongs to $\mathcal T_e$ exactly when this solution is nonnegative and
  \[
    |\{i:a\in W_{e_i}(z)\}|\geq n_a,
    \qquad
    |\{i:b\in W_{e_i}(z)\}|\geq n_b.
  \]
  To see that the inequalities are sufficient, first assign every child edge with $|W_{e_i}(z)|=1$ its unique available weight.
  The first inequality ensures that at most $n_b$ of these edges receive weight $b$, and the second ensures that at most $n_a$ receive weight $a$.
  Assign weight $a$ to enough of the remaining child edges with $W_{e_i}(z)=\{a,b\}$ to obtain exactly $n_a$ edges of weight $a$, and assign weight $b$ to the others.

  If $v$ is a leaf, then $\mathcal T_e=\{(x,x):x\in W_e\}$.
  Processing the remaining tables from the leaves toward the root therefore computes all tables.
  For a vertex $v$, its full label has the form $d_T(v)a+t(b-a)$ for some $t\in\{0,\dots,d_T(v)\}$, so each table has $O(|V|)$ candidate pairs.
  For every child edge $e_i$, all sets $W_{e_i}(z)$ can be computed for all $z\in\{d_T(v)a+t(b-a):0\leq t\leq d_T(v)\}$ by scanning $\mathcal T_{e_i}$, using $O(|V|^2)$ time per edge.
  Once these sets are available, all membership tests at a vertex take $O(|V|d_T(v))$ time.
  Summing over the tree gives an $O(|V|^3)$ algorithm.

  Let $e_0=v_0v_1$ be the unique edge incident with the root.
  The root has label $w(e_0)=x$, while a pair $(z,x)\in\mathcal T_{e_0}$ records the label $z$ of $v_1$.
  Hence a proper extension exists if and only if $\mathcal T_{e_0}$ contains a pair $(z,x)$ with $z\neq x$.
  Backtracking through the selected table entries recovers such an extension.
\end{proof}

The same dynamic program handles a minimum-cost version in which assigning an allowed weight to an edge has a cost.
For each entry $(z,x)\in\mathcal T_e$, store the minimum cost of a realization, including the cost of assigning weight $x$ to $e$.
For a fixed candidate pair $(z,x)$ at a vertex $v$, let $n_a$ and $n_b$ be the unique integers, if they exist, satisfying $n_a+n_b=d_T(v)-1$ and $a n_a+b n_b=z-x$.
If no such nonnegative integers exist, discard the candidate pair.
If some $W_{e_i}(z)$ is empty or either of the two availability inequalities above fails, discard the candidate pair.
For a child edge $e_i$ and an available weight $x_i\in W_{e_i}(z)$, let $c_i(x_i)$ be the minimum stored cost of an entry $(z_i,x_i)\in\mathcal T_{e_i}$ with $z_i\neq z$.
Assign the unique available weight to every child edge with $|W_{e_i}(z)|=1$.
Among the remaining child edges, all of which satisfy $W_{e_i}(z)=\{a,b\}$, assign weight $a$ to as many edges as needed to obtain $n_a$ child edges of weight $a$ in total, choosing those with the smallest values of $c_i(a)-c_i(b)$, and assign weight $b$ to the others.

\subsection{The antifactor problem}

The antifactor problem is the special case of the degree-prescribed subgraph problem in which exactly one degree is forbidden at each vertex.
An instance consists of a connected graph $G=(V,E)$ and a function $f$ satisfying $f(v)\in\{0,\dots,d_G(v)\}$ for every $v\in V$; the task is to find a spanning subgraph $H$ of $G$ such that $d_H(v)\neq f(v)$ for every $v\in V$.
Lov\'{a}sz gave a polynomial-time algorithm for this problem~\cite{lovasz1973antifactors}; later generalizations appear in~\cite{generalFactorsOfGraphs,FrankDegreeConstrained}.
He also proved that every instance whose graph contains a cycle is feasible~\cite{lovasz1973antifactors}.
We therefore focus on trees.

The next reduction shows that the antifactor problem on trees can also be solved by the dynamic program from \cref{thm:treet}.

\begin{theorem}\label{thm:antifactor}
  The antifactor problem for trees can be reduced in polynomial time to deciding whether a given tree has a proper $\{0,1\}$-edge-weighting.
\end{theorem}
\begin{proof}
  Consider an antifactor instance $(G,f)$ in which $G=(V,E)$ is a tree on $n$ vertices.
  We construct a tree $G'$ that has a proper $\{0,1\}$-edge-weighting if and only if the antifactor instance is feasible.

  Write $V=\{v_1,\dots,v_n\}$ and set $M=n+2$.
  Begin with a copy of $G$.
  For every $v_i$, add a set $D_i$ of $iM$ new neighbors of $v_i$.
  Choose one vertex $u_i\in D_i$, and attach $iM+f(v_i)-1$ pairwise internally vertex-disjoint paths of length two to $u_i$.
  For every vertex $u\in D_i\setminus\{u_i\}$, add one new leaf adjacent to $u$.
  \Cref{fig:application} illustrates the construction at $v_i$.
  The resulting graph $G'$ is a tree with $O(n^3)$ vertices and can be constructed in polynomial time.

\begin{figure}[t]
\centering
\begin{sideways}
\centering

\begin{tikzpicture}[xscale=.8,yscale=-1]
\SetVertexMath
\tikzset{VertexStyle/.append style = {minimum size = 8pt,inner sep=0pt}}

\begin{scope}[xscale=.45]
\begin{scope}[yscale=.8]
\draw[thick] (1.5/.45,.0) circle (2.5);
\end{scope}
\draw[thick] (-.5,-1.35) arc (225-12.5:135+12.5:2.5);
\end{scope}

\draw (1.5,-2.3) node[rotate=270] {$G$};

\Vertex[x=1.5,y=0,L=$ $]{dots}
\draw (1.5,-.3) node[rotate=270] {$v_i$};

\Vertex[x=-1.5,y=2,L=$ $]{dots1}
\Vertex[x=-1.5,y=-.5,L=$ $]{dots2}
\Vertex[x=-1.5,y=-2,L=$ $]{dots3}

\draw (-1.5,-2.3) node[rotate=270] {$u_i\ $};

\Vertex[x=-3,y=-3,L=$ $]{dots7}
\Vertex[x=-3,y=-1,L=$ $]{dots8}

\Vertex[x=-4.5,y=-4,L=$ $]{dots13}
\Vertex[x=-4.5,y=0,L=$ $]{dots14}

\Vertex[x=-4.5+1.5,y=4-1,L=$ $]{dots19}
\Vertex[x=-1.5-1.5,y=-.5+1,L=$ $]{dots20}

\draw[] (dots19)--(dots1);
\draw[] (dots20)--(dots2);

\tikzset{VertexStyle/.style = {shape = circle,fill = black,minimum size = 2pt,inner sep=0pt}}
\Vertex[x=-1.5,y=.4,L=$ $]{dots4}
\Vertex[x=-1.5,y=.75,L=$ $]{dots5}
\Vertex[x=-1.5,y=1.1,L=$ $]{dots6}

\Vertex[x=-3,y=1.4,L=$ $]{dots21}
\Vertex[x=-3,y=1.75,L=$ $]{dots22}
\Vertex[x=-3,y=2.1,L=$ $]{dots23}

\draw[thick] (-1,2.2) arc (0:180:.5);
\draw[thick] (-2,-2.2) arc (180:360:.5);
\draw[thick] (-2,2.2) -- (-2,-2.2);
\draw[thick] (-1,2.2) -- (-1,-2.2);
\draw (-1.5,-3) node[rotate=270] {$D_i$};

\draw[] (dots)--(dots1);
\draw[] (dots)--(dots2);
\draw[] (dots)--(dots3);

\path [draw half paths={draw=none}{dotted}] (dots4)--(dots);
\path [draw half paths={draw=none}{dotted}] (dots5)--(dots);
\path [draw half paths={draw=none}{dotted}] (dots6)--(dots);

\path [draw half paths={draw=none}{dotted}] (dots21)--(dots4);
\path [draw half paths={draw=none}{dotted}] (dots22)--(dots5);
\path [draw half paths={draw=none}{dotted}] (dots23)--(dots6);

\tikzset{VertexStyle/.style = {shape = circle,fill = black,minimum size = 2pt,inner sep=0pt}}
\Vertex[x=-3,y=-2+0.3125,L=$ $]{dots10}
\Vertex[x=-3,y=-2,L=$ $]{dots11}
\Vertex[x=-3,y=-2-0.3125,L=$ $]{dots12}

\draw[] (dots3)--(dots7);
\draw[] (dots3)--(dots8);

\path [draw half paths={draw=none}{dotted}] (dots10)--(dots3);
\path [draw half paths={draw=none}{dotted}] (dots11)--(dots3);
\path [draw half paths={draw=none}{dotted}] (dots12)--(dots3);


\tikzset{VertexStyle/.style = {shape = circle,fill = black,minimum size = 2pt,inner sep=0pt}}
\Vertex[x=-4.5,y=-2+2*0.3125,L=$ $]{dots16}
\Vertex[x=-4.5,y=-2,L=$ $]{dots17}
\Vertex[x=-4.5,y=-2-2*0.3125,L=$ $]{dots18}

\draw[] (dots13)--(dots7);
\draw[] (dots14)--(dots8);

\path [draw half paths={draw=none}{dotted}] (dots16)--(dots10);
\path [draw half paths={draw=none}{dotted}] (dots17)--(dots11);
\path [draw half paths={draw=none}{dotted}] (dots18)--(dots12);

\draw[decorate,decoration={brace,amplitude=3pt}]
(-4.75,.11) node(t_k_unten){} --
(-4.75,-4.11) node(t_k_opt_unten){};
\node[rotate=270] at (-5.15,-2.0){$iM+f(v_i)-1$};

\draw (-.15,-1.6) node[rotate=270] {$iM$};

\end{tikzpicture}

\end{sideways}
\caption{Illustration of the construction described in the proof of \cref{thm:antifactor}.}\label{fig:application}
\end{figure}
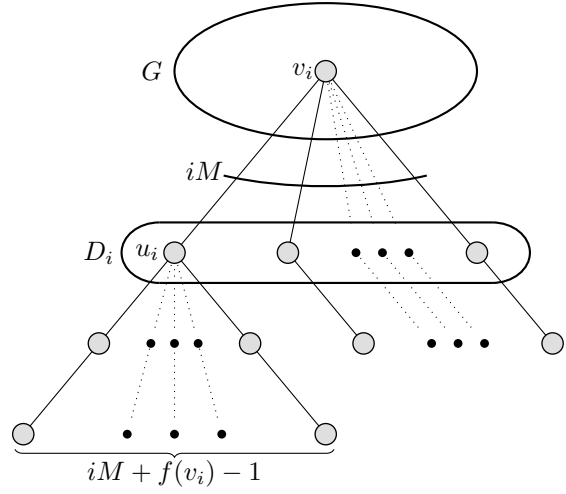

  Suppose first that the antifactor instance has a feasible solution $H$.
  Assign weight $1$ to an original edge $e\in E$ if $e\in E(H)$ and weight $0$ otherwise, and assign weight $1$ to every newly added edge.
  We verify that the resulting $\{0,1\}$-edge-weighting of $G'$ is proper.

  For every original vertex $v_i$, its label lies between $iM$ and $iM+n-1$.
  Since $(i+1)M-(iM+n-1)=M-n+1=3$, the labels of the original vertices are pairwise distinct, and every original edge is proper.
  Each vertex in $D_i\setminus\{u_i\}$ has label $2$, while its leaf neighbor has label $1$ and $z(v_i)\geq iM\geq M>2$.
  Moreover, $z(v_i)=iM+d_H(v_i)$ and $z(u_i)=iM+f(v_i)$, so the edge $v_iu_i$ is proper because $d_H(v_i)\neq f(v_i)$.
  Every internal vertex of a length-two path attached to $u_i$ has label $2$, its leaf endpoint has label $1$, and $z(u_i)\geq M>2$.
  Thus every edge of $G'$ is proper.

  Conversely, suppose that $G'$ has a proper $\{0,1\}$-edge-weighting, and let $H$ be the spanning subgraph of $G$ formed by the original edges of weight $1$.
  For each $u\in D_i\setminus\{u_i\}$, the edge $uv_i$ must have weight $1$; otherwise, $u$ and its leaf neighbor would have equal labels.
  Similarly, every edge incident with $u_i$ other than $u_iv_i$ must have weight $1$, because if the edge from $u_i$ to the internal vertex of an attached length-two path had weight $0$, that internal vertex and the leaf endpoint of the path would have equal labels.
  Therefore
  \[
    z(v_i)=d_H(v_i)+iM-1+w(u_iv_i),
    \qquad
    z(u_i)=iM+f(v_i)-1+w(u_iv_i).
  \]
  Since $u_iv_i$ is proper, these labels are distinct, and hence $d_H(v_i)\neq f(v_i)$ for every $i$.
  Thus $H$ is a feasible solution of the antifactor instance.
\end{proof}

\section*{Acknowledgment}
The work was supported by the Lend\"ulet Programme of the Hungarian Academy of Sciences -- grant number LP2021-1/2021, and by the Ministry of Innovation and Technology of Hungary from the National Research, Development and Innovation Fund, financed under the ELTE TKP 2021-NKTA-62 funding scheme.

\section*{Conflict of interest}
The authors declare no potential conflict of interest.

\bibliographystyle{plain}
\bibliography{bibliography}

\end{document}